\documentclass{CSML}
\pdfoutput=1

\usepackage{lastpage}
\newcommand{\dOi}{14(2:12)2018}
\lmcsheading{}{1--\pageref{LastPage}}{}{}%
{May~23, 2017}{May~22, 2018}{}

\usepackage{latexsym,xspace,calc,amsthm}
\usepackage{amsmath,amssymb} 

\usepackage{enumitem,mdwlist} 

\makecompactlist{enumeratei*}{enumeratei}

\usepackage{xparse}
\usepackage{tikz}
\usetikzlibrary{calc}
\DeclareDocumentCommand{\hcancel}{mO{0pt}O{1pt}O{0pt}O{-1pt}}{%
    \tikz[baseline=(tocancel.base)]{
        \node[inner sep=0pt,outer sep=0pt] (tocancel) {#1};
        \draw[gray] ($(tocancel.south west)+(#2,#3)$) -- ($(tocancel.north east)+(#4,#5)$);
    }%
}%

\usepackage{tgtermes}
\usepackage{fix-cm}

\usepackage{float}
\floatstyle{boxed} \restylefloat{figure}

\usepackage[all]{xy}

\let\myfresh\#
\def\#{\ensuremath{\text{\tt\myfresh}}}

\newcommand\mjg[1]{} 

\newcommand\figref[2]{\text{Figure~\ref{#1}} \rulefont{#2}}

\newcommand\integer{{\mathbb Z}}

\newcommand\minus{{\text{-}}}
\newcommand\plus{{\text{+}}}

\newcommand\minlev{\f{minlevel}}

\newcommand\finsubseteq{\mathbin{\subseteq_{\text{\it fin}}}}

\newcommand\level{\f{level}}

\newcommand\age{\f{age}}

\newcommand\compressthis[1]{\pmb{\hspace{.8pt}\raisebox{.5pt}{\scalebox{.85}{$#1$}}\hspace{.2pt}}}

\newcommand\tneg{{\pmb\neg}}

\newcommand\tbot{{\pmb\bot}}
\newcommand\teq{{\pmb{\text{=}}}}
\newcommand\tand{{\pmb\wedge}}

\newcommand\tin{{\pmb{\in}}}
\newcommand\tst[2]{\compressthis{\{}#1\hspace{.25pt}\compressthis{\mid}\hspace{.4pt}#2\compressthis{\}}}

\newcommand\tall{{\compressthis{\forall}}}

\newcommand\texi{{\pmb\exists}}

\renewcommand\land{\wedge}

\newcommand\limp{\Rightarrow}
\newcommand\mathsc[1]{\text{\tt #1}}

\newcommand\crediff{{\mathsc{iff}}}
\newcommand\myiff[2]{\crediff(#1,#2)} 
\newcommand\credimp{{\mathsc{imp}}}
\newcommand\myimp[2]{\credimp(#1,#2)}

\newcommand\credand{{\mathsc{and}}}
\newcommand\credor{{\mathsc{or}}}

\newcommand\credelt{{\mathsc{elt}}}
\newcommand\credneg{{\mathsc{neg}}}

\newcommand\credempset{{\mathsc{empt}}}
\newcommand\credtrue{{\mathsc{T}}}
\newcommand\credfalse{{\mathsc{F}}}
\newcommand\credfullset{{\mathsc{set}}}

\newcommand\credall{{\mathsc{all}}} 
\newcommand\credst{{\mathsc{st}}} 
 
\newcommand\atm{\mathsc{atm}}
\newcommand\credatm{\mathsc{atm}}
\newcommand\myneg[1]{\credneg(#1)}
\newcommand\myand[1]{{\credand}(#1)}
\newcommand\mybinaryand[2]{\myand{\{#1,#2\}}}
\newcommand\myor[1]{{\credor}(#1)}
\newcommand\myelt[2]{\credelt(#1,#2)}
\newcommand\myall[2]{\credall([#1]#2)}
\newcommand\myatm[1]{\credatm(#1)}
\newcommand\myst[2]{\mathsc{st}([#1]#2)}

\usepackage{upgreek}

\usepackage{yfonts}

\newcommand\Func{{\Rightarrow}}

 \renewenvironment{thebibliography}[1]{%
   \begin{odlthebibliography}{#1}%
     \setlength{\parskip}{0ex}%
     \setlength{\itemsep}{3pt}%
     \fontsize{10}{10} 
     \selectfont
}%
 {%
   \end{odlthebibliography}%
 }

\usepackage{hyperref}
\hypersetup{hidelinks}

\usepackage{fancybox}
\newlength{\mylength}
{\setlength{\fboxsep}{5pt}
\setlength{\mylength}{\linewidth}%
\addtolength{\mylength}{-2\fboxsep}%
\addtolength{\mylength}{-2\fboxrule}%
\Sbox
\minipage{\mylength}%
\setlength{\abovedisplayskip}{0pt}%
\setlength{\belowdisplayskip}{0pt}%
$$}%
{$$\endminipage\endSbox
{\setlength{\abovedisplayskip}{1pt}%
\setlength{\belowdisplayskip}{0pt}%
\[\fbox{\TheSbox}\]}}

\newenvironment{frametxt}%
{\setlength{\fboxsep}{5pt}
\setlength{\mylength}{\linewidth}%
\addtolength{\mylength}{-2\fboxsep}%
\addtolength{\mylength}{-2\fboxrule}%
\Sbox
\minipage{\mylength}%
\setlength{\abovedisplayskip}{5pt}%
\setlength{\belowdisplayskip}{5pt}%
}%
{\endminipage\endSbox
{\setlength{\abovedisplayskip}{1pt}%
\setlength{\belowdisplayskip}{0pt}%
\[\fbox{\TheSbox}\]}}

\usepackage{graphics}
\usepackage{ifthen}
\usepackage{verbatim}
\newdimen\proofrulebreadth \proofrulebreadth=.05em
\newdimen\proofdotseparation \proofdotseparation=1.25ex
\newdimen\proofrulebaseline \proofrulebaseline=2ex
\newcount\proofdotnumber \proofdotnumber=3
\let\then\relax
\def\hfi{\hskip0pt plus.0001fil}
\mathchardef\squigto="3A3B
\newif\ifinsideprooftree\insideprooftreefalse
\newif\ifonleftofproofrule\onleftofproofrulefalse
\newif\ifproofdots\proofdotsfalse
\newif\ifdoubleproof\doubleprooffalse
\let\wereinproofbit\relax
\newdimen\shortenproofleft
\newdimen\shortenproofright
\newdimen\proofbelowshift
\newbox\proofabove
\newbox\proofbelow
\newbox\proofrulename
\def\shiftproofbelow{\let\next\relax\afterassignment\setshiftproofbelow\dimen0 }
\def\shiftproofbelowneg{\def\next{\multiply\dimen0 by-1 }%
\afterassignment\setshiftproofbelow\dimen0 }
\def\setshiftproofbelow{\next\proofbelowshift=\dimen0 }
\def\setproofrulebreadth{\proofrulebreadth}

\def\prooftree{
\ifnum  \lastpenalty=1
\then   \unpenalty
\else   \onleftofproofrulefalse
\fi
\ifonleftofproofrule
\else   \ifinsideprooftree
        \then   \hskip.5em plus1fil
        \fi
\fi
\bgroup
\setbox\proofbelow=\hbox{}\setbox\proofrulename=\hbox{}%
\let\justifies\proofover\let\leadsto\proofoverdots\let\Justifies\proofoverdbl
\let\using\proofusing\let\[\prooftree
\ifinsideprooftree\let\]\endprooftree\fi
\proofdotsfalse\doubleprooffalse
\let\thickness\setproofrulebreadth
\let\shiftright\shiftproofbelow \let\shift\shiftproofbelow
\let\shiftleft\shiftproofbelowneg
\let\ifwasinsideprooftree\ifinsideprooftree
\insideprooftreetrue
\setbox\proofabove=\hbox\bgroup$\displaystyle 
\let\wereinproofbit\prooftree
\shortenproofleft=0pt \shortenproofright=0pt \proofbelowshift=0pt
\onleftofproofruletrue\penalty1
}

\def\eproofbit{
\ifx    \wereinproofbit\prooftree
\then   \ifcase \lastpenalty
        \then   \shortenproofright=0pt  
        \or     \unpenalty\hfil         
        \or     \unpenalty\unskip       
        \else   \shortenproofright=0pt  
        \fi
\fi
\global\dimen0=\shortenproofleft
\global\dimen1=\shortenproofright
\global\dimen2=\proofrulebreadth
\global\dimen3=\proofbelowshift
\global\dimen4=\proofdotseparation
\global\count255=\proofdotnumber
$\egroup  
\shortenproofleft=\dimen0
\shortenproofright=\dimen1
\proofrulebreadth=\dimen2
\proofbelowshift=\dimen3
\proofdotseparation=\dimen4
\proofdotnumber=\count255
}

\def\proofover{
\eproofbit 
\setbox\proofbelow=\hbox\bgroup 
\let\wereinproofbit\proofover
$\displaystyle
}%
\def\proofoverdbl{
\eproofbit 
\doubleprooftrue
\setbox\proofbelow=\hbox\bgroup 
\let\wereinproofbit\proofoverdbl
$\displaystyle
}%
\def\proofoverdots{
\eproofbit 
\proofdotstrue
\setbox\proofbelow=\hbox\bgroup 
\let\wereinproofbit\proofoverdots
$\displaystyle
}%
\def\proofusing{
\eproofbit 
\setbox\proofrulename=\hbox\bgroup 
\let\wereinproofbit\proofusing
\kern0.3em$
}

\def\endprooftree{
\eproofbit 
  \dimen5 =0pt
\dimen0=\wd\proofabove \advance\dimen0-\shortenproofleft
\advance\dimen0-\shortenproofright
\dimen1=.5\dimen0 \advance\dimen1-.5\wd\proofbelow
\dimen4=\dimen1
\advance\dimen1\proofbelowshift \advance\dimen4-\proofbelowshift
\ifdim  \dimen1<0pt
\then   \advance\shortenproofleft\dimen1
        \advance\dimen0-\dimen1
        \dimen1=0pt
        \ifdim  \shortenproofleft<0pt
        \then   \setbox\proofabove=\hbox{%
                        \kern-\shortenproofleft\unhbox\proofabove}%
                \shortenproofleft=0pt
        \fi
\fi
\ifdim  \dimen4<0pt
\then   \advance\shortenproofright\dimen4
        \advance\dimen0-\dimen4
        \dimen4=0pt
\fi
\ifdim  \shortenproofright<\wd\proofrulename
\then   \shortenproofright=\wd\proofrulename
\fi
\dimen2=\shortenproofleft \advance\dimen2 by\dimen1
\dimen3=\shortenproofright\advance\dimen3 by\dimen4
\ifproofdots
\then
        \dimen6=\shortenproofleft \advance\dimen6 .5\dimen0
        \setbox1=\vbox to\proofdotseparation{\vss\hbox{$\cdot$}\vss}%
        \setbox0=\hbox{%
                \advance\dimen6-.5\wd1
                \kern\dimen6
                $\vcenter to\proofdotnumber\proofdotseparation
                        {\leaders\box1\vfill}$%
                \unhbox\proofrulename}%
\else   \dimen6=\fontdimen22\the\textfont2 
        \dimen7=\dimen6
        \advance\dimen6by.5\proofrulebreadth
        \advance\dimen7by-.5\proofrulebreadth
        \setbox0=\hbox{%
                \kern\shortenproofleft
                \ifdoubleproof
                \then   \hbox to\dimen0{%
                        $\mathsurround0pt\mathord=\mkern-6mu%
                        \cleaders\hbox{$\mkern-2mu=\mkern-2mu$}\hfill
                        \mkern-6mu\mathord=$}%
                \else   \vrule height\dimen6 depth-\dimen7 width\dimen0
                \fi
                \unhbox\proofrulename}%
        \ht0=\dimen6 \dp0=-\dimen7
\fi
\let\doll\relax
\ifwasinsideprooftree
\then   \let\VBOX\vbox
\else   \ifmmode\else$\let\doll=$\fi
        \let\VBOX\vcenter
\fi
\VBOX   {\baselineskip\proofrulebaseline \lineskip.2ex
        \expandafter\lineskiplimit\ifproofdots0ex\else-0.6ex\fi
        \hbox   spread\dimen5   {\hfi\unhbox\proofabove\hfi}%
        \hbox{\box0}%
        \hbox   {\kern\dimen2 \box\proofbelow}}\doll%
\global\dimen2=\dimen2
\global\dimen3=\dimen3
\egroup 
\ifonleftofproofrule
\then   \shortenproofleft=\dimen2
\fi
\shortenproofright=\dimen3
\onleftofproofrulefalse
\ifinsideprooftree
\then   \hskip.5em plus 1fil \penalty2
\fi
}

\newcommand\ns{\mathsf}

\newcommand\theory[1]{\ensuremath{\mathsf{#1}}\xspace}

\def\:{{\hspace{-1pt}{:}\hspace{-1.25pt}{:}\hspace{-.5pt}}}
\usepackage{relsize}

\def\id{\mathtxt{id}}

\newcommand\ssm{{{:}\text{=}}}

\newcommand\deffont[1]{{\bf #1}}

\newcommand\mone{{{\text{-}1}}}
\newcommand\liff{\Leftrightarrow}

\newcommand\size{\f{size}}
\newcommand\complexity{\f{cplx}}

\newcommand\supp{\f{supp}}

\newcommand\f[1]{\mathit{#1}}

\newcommand\atoms{\ensuremath{\mathbb{A}}\xspace}

\newcommand\dact[1]{}

\newcommand\atomsi[1]{\atoms^{\hspace{-1.7pt}{#1}}}
\newcommand\atomsj[1]{\atoms^{\hspace{-1pt}{#1}}}

\newcommand\powerset{\f{pset}}

\newcommand\finte[1]{{\langle #1 \rangle}}
\newcommand\vect[1]{\overline{#1}}

\newcommand\act[0]{{\cdot}}

\newcommand{\Defiff}
 {\mathrel{\ \ \stackrel{\scriptstyle \mathrm{def}}{\Leftrightarrow}\ \ }}

\newcommand{\defeq}
  {\stackrel{\mathrm{def}}{\,=\,}}

\newcommand\fix{\f{fix}}

\newcounter{jamieitemcounter}

\newenvironment{thrm}{\begin{thm}}{\end{thm}}
\newenvironment{lemm}{\begin{lem}}{\end{lem}}
\newenvironment{corr}{\begin{cor}}{\end{cor}}
\newenvironment{defn}{\begin{defi}}{\end{defi}}
\newenvironment{nttn}{\begin{nota}}{\end{nota}}
\newenvironment{xmpl}{\begin{exa}}{\end{exa}}
\newenvironment{rmrk}{\begin{rem}}{\end{rem}}

\newcommand\sm{{\mapsto}}

\usepackage{stmaryrd}

\newcommand\mathtxt[1]{ \ensuremath{\mathrm{#1}} }

\newcommand\rulefont[1]{\ensuremath{{\mathrm{\bf (#1)}}}}

\newcommand\Forall[1]{\forall #1.}
\newcommand\Exists[1]{\exists #1.}

\def\at{\text{@}}

\newcommand\cred{\theory{Pred}}

\newcommand\lset[1]{\theory{Set}^{#1}}

\usepackage{datetime}

\makeatletter
\@ifpackageloaded{fdsymbol}\@tempswafalse\@tempswatrue
\if@tempswa
  \newcommand{\fdsy@scale}{1.0}
  \newcommand\fdsy@mweight@normal{Book}
  \newcommand\fdsy@mweight@small{Book}
  \newcommand\fdsy@bweight@normal{Medium}
  \newcommand\fdsy@bweight@small{Medium}
  \DeclareFontFamily{U}{FdSymbolA}{}
  \DeclareSymbolFont{fdsymbols}{U}{FdSymbolA}{m}{n}%
  \SetSymbolFont{symbols}{bold}{U}{FdSymbolA}{b}{n}%
  \DeclareFontShape{U}{FdSymbolA}{m}{n}{
      <-7.1> s * [\fdsy@scale] FdSymbolA-\fdsy@mweight@small
      <7.1-> s * [\fdsy@scale] FdSymbolA-\fdsy@mweight@normal
  }{}
  \DeclareFontShape{U}{FdSymbolA}{b}{n}{
      <-7.1> s * [\fdsy@scale] FdSymbolA-\fdsy@bweight@small
      <7.1-> s * [\fdsy@scale] FdSymbolA-\fdsy@bweight@normal
  }{}
  \DeclareMathSymbol{\aleph}{\mathord}{fdsymbols}{"C7}
  \DeclareMathSymbol{\beth}{\mathord}{fdsymbols}{"C8}
  \DeclareMathSymbol{\gimel}{\mathord}{fdsymbols}{"C9}
  \DeclareMathSymbol{\daleth}{\mathord}{fdsymbols}{"CA}
\fi
\makeatother

\begin{document}

\title[Stratified Sets is confluent and normalising]{The language of Stratified Sets is confluent and strongly normalising} 

\author[M.~Gabbay]{Murdoch J. Gabbay}
\address{Heriot-Watt University, Scotland, UK}
\urladdr{www.gabbay.org.uk}
\thanks{Thanks to the editor and to the anonymous referees}

\begin{abstract}
We study the properties of the language of Stratified Sets (first-order logic with $\in$ and a stratification condition) as used in TST, TZT, and (with stratifiability instead of stratification) in Quine's NF.
We find that the syntax forms a nominal algebra for substitution and that stratification and stratifiability imply confluence and strong normalisation under rewrites corresponding naturally to $\beta$-conversion.
\end{abstract}

\maketitle

\section{Introduction}

\subsection{About Stratified Sets}

Consider Russell's paradox, that if $s=\{a\mid a\not\in a\}$ then $s\in s$ if and only if $s\not\in s$.
One way to avoid the term $s$ is to restrict to the language of \emph{Stratified Sets}.
This is first-order logic with a binary relation $t\tin s$ whose intuition is `$t$ is an element of $s$' and:
\begin{itemize}
\item
Variable symbols $a$ (called \emph{atoms} in this paper) are assigned levels, which are typically integers or natural numbers.
\item
We impose a \emph{stratification} typing condition that we may only form $t\tin s$ if the level of $s$ is one plus the level of $t$.
$\level(s)=\level(t)\plus 1$.
\end{itemize}
See Definition~\ref{defn.levels} for full details.
Then $s=\{a\mid a\not\in a\}$ 
cannot be stratified, since whatever level we assign to $a$ in $a\in a$ we cannot make $\level(a)=\level(a)\plus 1$.

Stratified Sets are one of a family of syntaxes designed to exclude Russell's paradox:
\begin{itemize*}
\item
The language of ZF set theory restricts sets comprehension to bounded comprehension $\{a\in X\mid \phi\}$.
\item
Type Theories (such as Higher-Order Logic) impose more or less elaborate type systems.
The canonical example of this is \emph{simple types} $\tau ::= \iota \mid \tau\to\tau$.
\item
Stratified Sets stratifies terms as described. 
\end{itemize*} 
One feature of Stratified Sets is that we can write a term representing the universal set:
$$
\mathbf{univ}=\{a\mid \top\}
$$ 
is easily stratified by giving $a$ any level we like.
Likewise we can write definitions such as `the number 2' to be `the set of all two-element sets':
$$
\mathbf 2=\{a\mid \Exists{b,c}(a=\{b,c\}\land b\neq c)\} 
$$
(Here we freely use syntactic sugar for readability; this can all be made fully formal.)

This feels liberating: 
we have the pleasure of full unbounded sets comprehension\footnote{It still needs to be stratified, of course, but by Russell's paradox we must expect our party to be spoiled. Our choices are only: how, and where?} 
and we have the pleasure of more sweeping types than are possible in the usual type theories such as Higher-Order Logic and its elaborations.\footnote{Two attitudes are possible with types: embrace and enrich them, which leads us in the direction of (for instance) dependent types, or minimise type structure.  Stratification minimises types all the way down to just being `$i\in\mathbb Z$'.}

\subsection{What this paper does}

The published literature using Stratified Sets does not view the basic syntax from the point of view of rewriting.
On this topic, this paper makes three observations:
\begin{enumerate}
\item
The stratification condition implies that the syntax is confluent and strongly normalising under the natural rewrite
$$
t\in \{a\mid \phi\}\quad \to\quad \phi[a\ssm t] .
$$
We can write: 
\begin{quote}
\emph{Stratification}\ \ $\limp$\ \ \emph{confluence and strong normalisation}.
\end{quote}
Similarly for stratifiability.
See Theorems~\ref{thrm.tst.strong} and~\ref{thrm.nf}.
\item
The syntax of normal forms 
becomes an algebra for substitution in a sense that will be made formal using nominal algebra.

See Theorem~\ref{thrm.capasn}; in fact the proof of Theorem~\ref{thrm.tst.strong} uses this.
\item
Our proof is constructed using nominal techniques.
The proofs in this paper should be fairly directly implementable in a nominal theorem-prover, such as Nominal Isabelle \cite{urban:nomrti}.
\end{enumerate}
In some senses, this paper is deliberately conventional, even simple: we write down a syntax and a rewrite relation and prove some nice properties.
But the simplicity is deceptive:
\begin{itemize*}
\item
TST, TZT, and NF as usually presented do not include sets comprehension in their syntax, if that syntax is even made fully formal.
So just noting that there are rewrite relations here that might be useful to look at, seems to be a new observation.
\item
The proofs are not trivial.
It is easy to give a handwaving argument (such as that given in the first half of Remark~\ref{rmrk.lambda} below) but it is surprisingly difficult to give a rigorous proof with all details.  
\end{itemize*}
More on this in Section~\ref{sect.conclusions}.

We use nominal techniques (see the material in Section~\ref{sect.basic.defs}) to manage the $\alpha$-binding in the syntax for universal quantification and sets comprehension. 
If the reader is unfamiliar with nominal techniques then they can just ignore this aspect: wherever we see reference to a nominal theorem, we can replace it with `by $\alpha$-conversion' or with `it is a fact of syntax that'.
The result should then be close to the kind of argument that might normally pass 
without comment or challenge.

\subsection{Some remarks}

\begin{rmrk}
\label{rmrk.lanss}
`The language of Stratified Sets' is a description specific to this paper.
In the literature, this syntax is unnamed and presented along with the theory we express using it: 
\begin{enumerate}
\item
TST (which stands for Typed Set Theory) is typically taken to be first-order logic with $\in$ and variables stratified as $\mathbb N=\{0,1,2,\dots\}$, along with reasonable axioms for first-order logic and extensional sets equality. 
\item
TZT is typically taken to be as the syntax and axioms of TST but with variables stratified as $\mathbb Z=\{0,1,\minus 1,2,\minus 2,\dots\}$.
\item
Quine's New Foundations (NF) uses the language of first-order logic with $\in$, and reasonable axioms, and a \deffont{stratifiability condition} that variables \emph{could} be stratified.

So in TST we would have to write (say) $a^1\in b^2$ (choosing a level 1 variable symbol $a$ and a level 2 variable symbol $b$), whereas in NF we could write just $a\in b$ and say ``we \emph{could} assign $a$ level 1 and $b$ level 2''.
\end{enumerate}
\end{rmrk}

\begin{rmrk}
\label{rmrk.ambiguity}
There is a slight ambiguity when we talk about stratification whether we insist that the syntax come delivered with an assignment of levels to all terms, or whether we insist on the weaker condition that an assignment \emph{could} be made, but this assignment need not be a structural part of the formula or term.
This distinguishes the languages of TST and TZT from that of NF: TST and TZT insist on stratification, and NF insists on stratifiability.

Our results will be agnostic in this choice (see for instance Theorems~\ref{thrm.tst.strong} and~\ref{thrm.nf}).
So when we write \emph{Stratified Sets} we could just as well write \emph{Stratifiable Sets} and everything would still work with only minor bookkeeping changes.
\end{rmrk}

\begin{rmrk}
\label{rmrk.protagonist}
The reader with a background in TST, TZT, and NF should note that this is not a paper about logical theories: it is a paper about their \emph{syntax}.
This is why we talk about `Stratified Sets' in this paper, and not e.g. `Typed Set Theories'.
Our protagonist is a language, not a logic. 
\end{rmrk}

\begin{rmrk}[Some words on terminology]
Some authors expand TST as `Theory of Simple Types'.
I think this terminology invites confusion with Simple Type Theory, so I prefer the alternative `Typed Set Theory'. 
I would also like to write that TZT stands for `Typed Zet Theory', but really TZT just stands for itself.
\end{rmrk}

\begin{rmrk}[References]
For the reader interested in the logical motivations for these syntaxes we provide references: 
\begin{itemize*}
\item
A historical account of Russell's paradox is in \cite{griffin:prehrp}. 
\item
For ZF set theory, see e.g. \cite{jech:sett}.
\item
Excellent discussions of TST, TZT, and NF are in \cite{forster:settus} and \cite{holmes:elestu}, and a clear summary with a brief but well-chosen bibliography is in \cite{forster:quinfs}.
\end{itemize*}
\end{rmrk}

\begin{rmrk}[Connection to the $\lambda$-calculus]
\label{rmrk.lambda}
One way to see that something \emph{like} this paper \emph{should} work, fingers crossed, is by an analogy:
\begin{itemize*}
\item
The rewrite $t\in \{a\mid \phi\} \to \phi[a\ssm t]$ can be rewritten as $(\lambda a.\phi)t\to\phi[a\ssm t]$.
\item
Extensionality is $s=\{b\mid b\in s\}$, and we can rewrite this as $s=\lambda b.(sb)$.
\end{itemize*}
These are of course familiar as $\beta$-reduction and $\eta$-expansion.
The proofs need to be checked 
but the analogy above invites an analysis of the kind that we will now carry out.

And indeed this has been done, though not for stratified sets.
In \cite{keller:hersst} (many thanks to an anonymous referee for bringing this to my attention) a development analogous to what is done in this paper for stratified syntax using nominal techniques, is carried out for the simply-typed lambda-calculus using de Bruijn indexes.
Definitions and results bear a very nice comparison: for instance, Lemma~6 of \cite{keller:hersst} corresponds to Proposition~\ref{prop.interp.sub}.\footnote{There are also differences:  
In \cite{keller:hersst} the authors implement their proof in Agda, whereas implementation of the proofs in this paper for future work.  
On the other hand, proofs in this paper are given more-or-less in full, whereas in \cite{keller:hersst} the authors elide technical details. 
This paper proves confluence and strong normalisation, which is strictly more than \cite{keller:hersst} which considers only the existence of a reduction path to normal forms.
}

A technical device in \cite{keller:hersst}, which goes back to a quite technical development in \cite{watkins:conlfp}, is to work directly with a datatype of normal forms and substitution on them; this is just like the \emph{internal syntax} and its \emph{sigma-action} that we will see in this paper.

For me the motivation for setting things up in this way is partly practical and partly abstract: internal syntax with its sigma-action turns out to be a nominal sigma-algebra (Theorem~\ref{thrm.capasn}) and this ties in with a literature on advanced nominal models of logic and computation \cite{gabbay:semooc,gabbay:repdul,gabbay:capasn-jv}.  
This paper was originally conceived as a prelude to building advanced nominal models of stratified and stratifiable syntaxes and type theories, though it has acquired independent interest.

It is therefore interesting to see analogous design choices appearing independently, motivated by apparently purely implementational concerns: what is good for the abstract mathematics also seems to be good for the proof-engineering.
\end{rmrk}

\section{Background on nominal techniques}
\label{sect.basic.defs}

Intuitively, a nominal set is ``a set $\ns X$ whose elements $x\in\ns X$ may `contain' finitely many names $a,b,c\in\atoms$''.
We may call names \emph{atoms}.
The notion of `contain' used here is not the obvious notion of `is a set element of': formally, we say that $x$ has \emph{finite support} (Definition~\ref{defn.supp}).

Nominal sets are formally defined in Subsection~\ref{subsect.basic.definitions}.
Examples are in Subsection~\ref{subsect.pow}.
The reader might prefer to read this section only briefly at first, and then use it as a reference for the later sections where these underlying ideas get applied. 
More detailed expositions are also in \cite{gabbay:newaas-jv,gabbay:fountl,gabbay:pernl-jv,pitts:nomsns}.

In the context of the broader literature, the message of this section is as follows:
\begin{itemize*}
\item
The reader with a category-theory background can read this section as exploring the category of nominal sets, or equivalently the Schanuel topos (more on this in \cite[Section III.9]{MLM:sgl},\ \cite[A.21, page 79]{johnstone:skeett},\ or \cite[Theorem~9.14]{gabbay:fountl}).
\item
The reader with a sets background can read this section as stating that we use Fraenkel-Mostowski set theory (\deffont{FM sets}).

A discussion of this sets foundation, tailored to nominal techniques, can be found in \cite[Section~10]{gabbay:fountl}.
FM sets add \emph{urelemente} or \emph{atoms} to the sets universe.
\item
The reader uninterested in foundations can note that previous work \cite{gabbay:newaas-jv,gabbay:fountl,gabbay:pernl-jv} has shown that just assuming names as primitive entities in Definition~\ref{defn.atoms} 
yields a remarkable clutch of definitions and results, including Theorem~\ref{thrm.supp}, Corollary~\ref{corr.stuff}, and Theorem~\ref{thrm.equivar}.
\end{itemize*}

\subsection{Basic definitions}
\label{subsect.basic.definitions}

\subsubsection{Atoms and permutations}
\label{subsect.cardinalities.and.atoms}

\begin{defn}
\label{defn.atoms}
\begin{itemize*}
\item
For each $i{\in}\integer$ fix a disjoint countably infinite set $\atomsi{i}$ of \deffont{atoms}.\footnote{These will serve as variable symbols in Definition~\ref{defn.cred}.} 
\item
Write $\atoms=\bigcup_{i{\in}\integer}\atomsi{i}$.
\item
If $a{\in}\atoms$ (so $a$ is an atom) write $\level(a)$ for the unique number such that $a{\in}\atomsi{\level(a)}$.
\item
We use a \deffont{permutative convention} that $a,b,c,\ldots$ range over \emph{distinct} atoms.

If we do not wish to use the permutative convention then we will refer to the atom using $n$ (see for instance \rulefont{\sigma\credelt\atm} of Figure~\ref{fig.sub}).
\end{itemize*}
\end{defn}

\subsubsection{Permutation actions on sets}

\begin{defn}
\label{defn.pi}
\label{defn.pol}
Suppose $\pi:\atoms\cong\atoms$ is a bijection on atoms.
\begin{enumerate}
\item
If $\f{nontriv}(\pi)=\{a\mid \pi(a)\neq a\}$ is finite then we call $\pi$ \deffont{finite}.
\item
If $\pi(a)\in\atomsi{i}\liff a\in\atomsi{i}$ then call $\pi$ \deffont{sort-respecting}.
\item
A \deffont{permutation} $\pi$ is a finite sort-respecting bijection on atoms.

Henceforth $\pi$ will range over permutations.
\end{enumerate}
\end{defn}

We will use the following notations in the rest of this paper:
\begin{nttn}
\begin{enumerate*}
\item
Write $\id$ for the \deffont{identity} permutation such that $\id(a)=a$ for all $a$.
\item
Write $\pi'\circ\pi$ for composition, so that $(\pi'\circ\pi)(a)=\pi'(\pi(a))$.
\item
If $i{\in}\integer$ and $a,b{\in}\atomsi{i}$ then write $(a\;b)$ for the \deffont{swapping} (terminology from \cite{gabbay:newaas-jv}) mapping $a$ to $b$,\ $b$ to $a$,\ and all other $c$ to themselves, and take $(a\;a)=\id$.
\item
Write $\pi^\mone$ for the inverse of $\pi$, so that $\pi^\mone\circ\pi=\id=\pi\circ\pi^\mone$.
\end{enumerate*}
\end{nttn}

\subsubsection{Sets with a permutation action}

\begin{nttn}
\label{nttn.fix}
If $A\subseteq\atoms$ write 
$$
\fix(A)=\{\pi\mid \Forall{a{\in} A}\pi(a)=a\}.
$$
\end{nttn}

\begin{defn}
\label{defn.set.with.perm}
A \deffont{set with a permutation action} $\ns X$ is a pair $(|\ns X|,\act)$ of an \deffont{underlying set} $|\ns X|$ and a \deffont{permutation action} written $\pi\act x$ 
 which is a group action on $|\ns X|$, so that $\id\act x=x$ and $\pi\act(\pi'\act x)=(\pi\circ\pi')\act x$ for all $x\in\ns X$ and permutations $\pi$ and $\pi'$.
\end{defn}

\begin{defn}
\begin{enumerate*}
\item
Say that $A\subseteq\atoms$ \deffont{supports} $x\in\ns X$ when $\Forall{\pi}\pi\in\fix(A)\limp \pi\act x=x$.
\item
If a finite $A\subseteq\atoms$ supporting $x$ exists, call $x$ \deffont{finitely supported} (by $A$) and say that $x$ has \deffont{finite support}.
\end{enumerate*}
\end{defn}

\begin{nttn}
If $\ns X$ is a set with a permutation action then we may write 
\begin{itemize*}
\item
$x\in\ns X$ as shorthand for $x\in|\ns X|$, and 
\item
$X\subseteq\ns X$ as shorthand for $X\subseteq|\ns X|$.
\end{itemize*}
\end{nttn}

\subsubsection{Nominal sets}

\begin{frametxt}
\begin{defn}
\label{defn.nominal.set}
Call a set with a permutation action $\ns X$ a \deffont{nominal set} when every $x\in\ns X$ has finite support.
$\ns X$, $\ns Y$, $\ns Z$ will range over nominal sets.
\end{defn}
\end{frametxt}

\begin{defn}
\label{defn.equivariant}
Call a function $f\in\ns X\Func\ns Y$ \deffont{equivariant} when $\pi\act (f(x))=f(\pi\act x)$ for all permutations $\pi$ and $x\in\ns X$.
In this case write $f:\ns X\Func\ns Y$.

The category of nominal sets and equivariant functions between them is usually called the category of \emph{nominal sets}.
\end{defn}

\begin{defn}
\label{defn.supp}
Suppose $\ns X$ is a nominal set and $x\in\ns X$.
Define the \deffont{support} of $x$ by
$$
\supp(x)=\bigcap\{A{\subseteq}\atoms \mid A\text{ is finite and supports }x\} .
$$
\end{defn}

\begin{nttn}
\label{nttn.fresh}
\begin{itemize*}
\item
Write $a\#x$ as shorthand for $a\not\in\supp(x)$ and read this as $a$ is \deffont{fresh for} $x$.
\item
If $T{\subseteq}\atoms$ write $T\#x$ as shorthand for $\Forall{a{\in}T}a\#x$.
\item
Given atoms $a_1,\dots,a_n$ and elements $x_1,\dots,x_m$ write $a_1,\dots,a_n\#x_1,\dots,x_m$ as shorthand for $\Forall{1{\leq}j{\leq}m}\{a_1,\dots,a_n\}\#x_j$.
That is: $a_i\#x_j$ for every $i$ and $j$.
\end{itemize*}
\end{nttn}

\begin{thrm}
\label{thrm.supp}
Suppose $\ns X$ is a nominal set and $x\in\ns X$.
Then $\supp(x)$ is the unique least finite set of atoms that supports $x$.
\end{thrm}
\begin{proof}
See \cite[Theorem~2.21(1)]{gabbay:fountl}.
\end{proof}

\begin{corr}
\label{corr.stuff}
\begin{enumerate*}
\item\label{stuff.fixsupp.fixelt}
If $\pi(a)=a$ for all $a\in\supp(x)$ then $\pi\act x=x$.
Equivalently:
\begin{enumerate*}
\item
If $\pi{\in}\fix(\supp(x))$ then $\pi\act x=x$.
\item 
If $\Forall{a{\in}\atoms}(\pi(a){\neq}a\limp a\#x)$ then $\pi\act x=x$ (see Notation~\ref{nttn.fresh}).
\end{enumerate*}
\item
If $\pi(a)=\pi'(a)$ for every $a{\in}\supp(x)$ then $\pi\act x=\pi'\act x$.
\item
$a\#x$ if and only if $\Exists{b}(b\#x\land (b\;a)\act x=x)$.
\end{enumerate*}
\end{corr}
\begin{proof}
By routine calculations from the definitions and from Theorem~\ref{thrm.supp} (see also \cite[Theorem~2.21(2)]{gabbay:fountl}).
\end{proof}

\subsection{Examples}
\label{subsect.pow}

Suppose $\ns X$ and $\ns Y$ are nominal sets.
We consider some examples, some of which will be useful later.

\subsubsection{Atoms}
\label{subsect.xmpl.atoms}

$\atoms$ is a nominal set with the \emph{natural permutation action} $\pi\act a=\pi(a)$.

\subsubsection{Cartesian product}
\label{subsect.cartesian.product}

$\ns X\times\ns Y$ is a nominal set with underlying set $\{(x,y)\mid x\in\ns X, y\in\ns Y\}$ and the \emph{pointwise} action $\pi\act(x,y)=(\pi\act x,\pi\act y)$.

It is routine to check that $\supp((x,y))=\supp(x){\cup}\supp(y)$.

\subsubsection{Full function space}
\label{subsect.full.function.space}

$\ns X{\to}\ns Y$ is a set with a permutation action with underlying set all functions from $|\ns X|$ to $|\ns Y|$, and the \deffont{conjugation} permutation action 
$$
(\pi\act f)(x)=\pi\act(f(\pi^\mone\act x)) .
$$

\subsubsection{Finite-supported function space}

$\ns X\Func\ns Y$ is a nominal set with underlying set the functions from $|\ns X|$ to $|\ns Y|$ with finite support under the conjugation action, and the conjugation permutation action.

\subsubsection{Full powerset}
\label{subsect.full.powerset}

\begin{defn}
\label{defn.pointwise.action}
Suppose $\ns Z$ is a set with a permutation action.
Give subsets $Z\subseteq\ns Z$ the \deffont{pointwise} permutation action
$$
\pi\act Z=\{\pi\act z\mid z\in Z\} .
$$
\end{defn}

Then $\powerset(\ns Z)$ (the full powerset of $\ns Z$) is a set with a permutation action with 
\begin{itemize*}
\item
underlying set $\{Z\mid Z\subseteq\ns Z\}$ (the set of all subsets of $|\ns Z|$), and 
\item
the pointwise action $\pi\act Z=\{\pi\act z\mid z\in Z\}$.
\end{itemize*}

A particularly useful instance of the pointwise action is for sets of atoms.
As discussed in Subsection~\ref{subsect.xmpl.atoms} above, if $a\in\atoms$ then $\pi\act a=\pi(a)$.
Thus if $A\subseteq\atoms$ then 
$$
\pi\act A\quad\text{means}\quad \{\pi(a)\mid a\in A\} .
$$

\begin{lemm}
Even if $\ns Z$ is a nominal set, $\powerset(\ns Z)$ need not be a nominal set.
\end{lemm}
\begin{proof}
Take $\ns Z=\atoms$ which we enumerate as $\{a_0,a_1,a_2,\dots\}$ and we take $Z\in\powerset(\ns Z)$ to be equal to $\f{comb}$ defined by 
$$
\f{comb}=\{a_0,a_2,a_4,\dots\} .
$$
This does not have finite support (see also \cite[Remark~2.18]{gabbay:fountl}). 
\end{proof}

\subsubsection{Finite powerset}
\label{subsect.finite.pow}

For this subsection, fix a nominal set $\ns X$.
\begin{defn}
\label{defn.finite.powerset}
Write $\f{FinPow}(\ns X)$ for the nominal set with
\begin{itemize*}
\item
underlying set the set of all finite subsets of $\ns X$,
\item
with the pointwise action from Definition~\ref{defn.pointwise.action}. 
\end{itemize*}
\end{defn}

\begin{nttn}
We might write $X\finsubseteq\ns X$ for $X{\in}\f{FinPow}(\ns X)$. 
\end{nttn}

\begin{lemm}
\label{lemm.finpow.support}
If $X\finsubseteq\ns X$ then:
\begin{enumerate*}
\item
$\bigcup\{\supp(x)\mid x{\in}X\}$ is finite. 
\item
\label{item.finpow.support.elts}
$\bigcup\{\supp(x)\mid x{\in}X\}=\supp(X)$. 
\item
\label{item.finpow.support.subset}
$x\in X$ implies $\supp(x)\subseteq\supp(X)$. 

Rewriting this using Notation~\ref{nttn.fresh}: if $X$ is finite and $x{\in}X$ then $a\#X$ implies $a\#x$.\footnote{This is not necessarily true if $X$ is infinite.  For instance if we take $\ns X=\atoms=X$ then the reader can verify that $a\#X$ for every $a$, but $a\#a$ does not hold for any $a\in \atoms$.  This is a feature of nominal techniques, not a bug; but for the case of finite sets, things are simpler.} 
\end{enumerate*}
\end{lemm}
\begin{proof}
The first part is immediate since by assumption there is some finite $A{\subseteq}\atoms$ that bounds $\supp(x)$ for all $x\in X$.
The second part follows by an easy calculation using part~3 of Corollary~\ref{corr.stuff}; full details are in \cite[Theorem~2.29]{gabbay:fountl}, of which Lemma~\ref{lemm.finpow.support} is a special case. 
Part~3 follows from the first and second parts.
\end{proof}

\subsubsection{Atoms-abstraction}
\label{subsect.atomsabs}

Atoms-abstraction was the first real application of nominal techniques; it was used to build inductive datatypes of syntax-with-binding.
Nominal atoms-abstraction captures the essence of $\alpha$-binding.
In this paper we use it to model the binding in universal quantification and sets comprehension (see Definition~\ref{defn.lsets}).
The maths here goes back to \cite{gabbay:thesis,gabbay:newaas-jv}; we give references to proofs in a more recent presentation \cite{gabbay:fountl}.

Assume a nominal set $\ns X$ and an $i{\in}\integer$.

\begin{defn}
\label{defn.abs}
Let the \deffont{atoms-abstraction} set $[\atomsi{i}]\ns X$ have
\begin{itemize*}
\item
Underlying set $\{[a]x\mid a{\in}\atomsi{i},\ x\in\ns X\}$ where $[a]x=\{(\pi(a),\pi\act x)\mid \pi\in\fix(\supp(x){\setminus}\{a\})\}$.
\item
Permutation action $\pi\act[a]x=[\pi\act a]\pi\act x$.
\end{itemize*}
\end{defn}

\begin{lemm}
\label{lemm.supp.abs}
If $x{\in}\ns X$ and $a{\in}\atomsi{i}$ then $\supp([a]x)=\supp(x){\setminus}\{a\}$.
In particular $a\#[a]x$ (Notation~\ref{nttn.fresh}).
\end{lemm}
\begin{proof}
See \cite[Theorem~3.11]{gabbay:fountl}.
\end{proof}

\begin{lemm}
\label{lemm.abs.alpha}
Suppose $x{\in}\ns X$ and $a,b{\in}\atomsi{i}$.
Then if $b\#x$ then $[a]x=[b](b\ a)\act x$.
\end{lemm}
\begin{proof}
See \cite[Lemma~3.12]{gabbay:fountl}.
\end{proof}

\begin{defn}
Suppose $z{\in}[\atomsi{i}]\ns X$ and $b{\in}\atomsi{i}$.
Write $z\at b$ for the unique $x\in\ns X$ such that $z=[b]x$, if this exists.
\end{defn}

\begin{lemm}
\label{lemm.abs.exists}
Suppose $b{\in}\atomsi{i}$ and $z\in[\atomsi{i}]\ns X$.
Then $b\#z$ implies $z\at b\in\ns X$ is well-defined. 
\end{lemm}
\begin{proof}
See \cite[Theorem~3.19]{gabbay:fountl}.
\end{proof}

\begin{lemm}
\label{lemm.abs.basic}
Suppose $a{\in}\atomsi{i}$ and $x{\in}\ns X$ and $z\in[\atomsi{i}]\ns X$.
Then:
\begin{enumerate*}
\item
$([a]x)\at a=x$ and if $b\#x$ then $([a]x)\at b=(b\ a)\act x$.
\item
If $a\#z$ then $[a](z\at a)=z$.
\end{enumerate*}
\end{lemm} 
\begin{proof}
See \cite[Theorem~3.19]{gabbay:fountl}.
\end{proof}

\subsection{The principle of equivariance}
\label{subsect.pre-equivar}

\begin{rmrk}
\label{rmrk.equivar.1}
We now come to the \emph{principle of equivariance} (Theorem~\ref{thrm.equivar}; see also \cite[Subsection~4.2]{gabbay:fountl} and \cite[Lemma~4.7]{gabbay:newaas-jv}).
It enables a particularly efficient management of renaming and $\alpha$-conversion in syntax and semantics and captures why it is so useful to use \emph{names} to model them instead of, for instance, numbers.

In a nutshell we can say
\begin{quote}
\emph{Atoms are distinguishable, but interchangable.}
\end{quote}
and we make this formal as follows: 
\end{rmrk}

\begin{thrm}
\label{thrm.equivar}
\label{thrm.no.increase.of.supp}
Suppose $\vect x$ is a list $x_1,\ldots,x_n$.
Suppose $\pi$ is a (not necessarily finite) permutation and write $\pi\act \vect x$ for $\pi\act x_1,\ldots,\pi\act x_n$.
Suppose $\Phi(\vect x)$ is a first-order logic predicate in the language of ZFA\footnote{First-order logic with equality $=$, sets membership $\in$, and a constant or collection of constants for sets of atoms.} with free variables $\vect x$.
Suppose $\Upsilon(\vect x)$ is a function specified using a first-order predicate in the language of ZFA with free variables $\vect x$.

Then we have the following principles:
\begin{enumerate*}
\item\label{equivar.pred}
\deffont{Equivariance of predicates.} \ $\Phi(\vect x) \liff \Phi(\pi\act \vect x)$.\footnote{It is important to realise here that $\vect x$ must contain \emph{all} the variables mentioned in the predicate.
It is not the case that $a=a$ if and only if $a=b$ --- but it is the case that $a=b$ if and only if $b=a$ (both are false).}
\item\label{equivar.fun}
\deffont{Equivariance of functions.}\quad\,$\pi\act \Upsilon(\vect x) = \Upsilon(\pi\act \vect x)$.
\item
\deffont{Conservation of support.}\quad\ \
If $\vect x$ denotes elements with finite support \\ then
$\supp(\Upsilon(\vect x)) \subseteq \supp(x_1){\cup}\cdots{\cup}\supp(x_n)$.
\end{enumerate*}
\end{thrm}
\begin{proof}
See Theorem~4.4, Corollary~4.6, and Theorem~4.7 from \cite{gabbay:fountl}.
\end{proof}

\begin{rmrk}
Theorem~\ref{thrm.equivar} states that atoms can be permuted in our theorems and lemmas provided we do so consistently in all parameters.
So for instance if we have proved $\phi(a,b,c)$, then 
\begin{itemize*}
\item
taking $\pi=(a\,c)$ we also know $\phi(c,b,a)$ and 
\item
taking $\pi=(a\,a')(b\,b')(c\,c')$ we also know $\phi(a',b',c')$, but 
\item
we do not necessarily know that we can deduce $\phi(a,b,a)$ (depending on $\phi$ this may still hold, of course, but not by equivariance since no permutation takes $(a,b,c)$ to $(a,b,a)$).
\end{itemize*}
Equivariance makes explicit a sense in which atoms have a dual nature: individually, atoms behave like pointers to themselves,\footnote{In the implementation of FM set theory in my PhD thesis \cite{gabbay:newaas-jv} this was literally true: I found it convenient to use \emph{Quine atoms}, meaning that $a=\{a\}$.} but collectively they have the flavour of variables ranging over the set of all atoms via the action of permutations.\footnote{This too can be made precise. See Subsection~2.6 and Lemma~4.17 of \cite{gabbay:pnlthf}.} 
See also the permutative convention from Definition~\ref{defn.atoms}.

We will use Theorem~\ref{thrm.equivar} frequently in this paper, either to move permutations around (parts~1 and~2) or to get `free' bounds on the support of elements (part~3).
`Free' here means `from the form of the definition, without having to verify it by calculations'.
Theorem~\ref{thrm.equivar} is `free' in the spirit of Wadler's marvellously titled \emph{Theorems for free!} \cite{wadler:theff}.\footnote{Finally, we can be somewhat more precise about the effort these free equivariance deductions can save:  With equivariance, the cost of deducing $\phi(\pi\act x,\pi\act y,\pi\act z)$, given a deduction of $\phi(x,y,z)$, is~1.  Without equivariance, the cost of deducing $\phi(\pi\act x,\pi\act y,\pi\act z)$, given a deduction of $\phi(x,y,z)$, is roughly $n$ where $n$ is the cost of deducing $\phi(x,y,z)$.  
This is convenient in a rigorous but unmechanised proof such as the one in this paper; in an implementation it can quadratically reduce effort by saving roughly effort $n$ for each $\phi$.  This is the difference between $\alpha$-equivalence and renaming lemmas being a minor consideration, and them inflating to dominate the development. My feeling is that once renaming lemmas consume more than 80\% of the developmental effort, development stalls.}

Discussions expanding on this remark are in~\cite{gabbay:equzfn} (full paper) and~\cite{gabbay:equzwc} (abstract).
\end{rmrk} 

\begin{prop}
\label{prop.pi.supp}
\begin{enumerate*}
\item
$\supp(\pi\act x)=\pi\act\supp(x)$ (which means $\{\pi(a)\mid a\in\supp(x)\}$).
\item
$a\#\pi\act x$ (Notation~\ref{nttn.fresh}) if and only if $\pi^\mone(a)\#x$, and $a\#x$ if and only if $\pi(a)\#\pi\act x$.
\end{enumerate*}
\end{prop}
\begin{proof}
Immediate consequence of part~2 of Theorem~\ref{thrm.equivar} (for the `not-free' proof by concrete calculations see \cite[Theorem~2.19]{gabbay:fountl}).
\end{proof}

\section{Internal syntax}

\subsection{Basic definition}
\label{subsect.internal.syntax}

\begin{rmrk}
We are now ready to to define our syntax (Figure~\ref{fig.int.syntax}) and study its basic properties (with more advanced properties considered in Section~\ref{sect.sigma-action}). 

Figure~\ref{fig.int.syntax} defines a \emph{nominal} datatype, in which atoms-abstraction is used to manage binding, as introduced in \cite{gabbay:newaas-jv}.
This gives us Lemma~\ref{lemm.tall.fresh}. 
\begin{enumerate*}
\item
Parts~\ref{tall.fresh.alpha.1} and~\ref{tall.fresh.alpha.2} of Lemma~\ref{lemm.tall.fresh} say ``We can alpha-convert''.
\item
In part~\ref{tall.fresh.supp} of Lemma~\ref{lemm.tall.fresh}, $\supp$ corresponds exactly to the notion that would normally be written ``Free variables of'', and $a\#X$ corresponds to ``$a$ is not free in $X$''.
\end{enumerate*}
So why not just write that?
Nominal techniques are a general basket of ideas with implications that go well beyond modelling syntax, but the specific benefit of using nominal techniques to model syntax is that we get alpha-conversion for free from the ambient nominal theory (see \cite{gabbay:newaas-jv} and Section~\ref{sect.basic.defs}).
We do not have to define $\alpha$-conversion and free variables of by induction, and then prove their properties (which is actually a more subtle undertaking than is often realised; cf. Remark~\ref{rmrk.substitution.lemma}).

The reader does not expect to see notions of ordered pairs, trees, numbers, functions, and function application developed from first principles every time we want to write abstract syntax and write a function on a syntax tree.
It is assumed that these things have been worked out.
Nominal techniques do that for binding (and more).
\end{rmrk}

\begin{nttn}
Write $\mathbb Z$ for the \deffont{integers}, so $\mathbb Z=\{0,1,\minus 1,2,\minus 2,\dots\}$ and $\mathbb N$ for the \deffont{natural numbers}, which we start at $0$, so $\mathbb N=\{0,1,2,\dots\}$.
\end{nttn}

\begin{defn}
\label{defn.cred}
\label{defn.lsets}
\label{defn.age}
\begin{enumerate*}
\item
Define datatypes 
\begin{itemize*}
\item
$\cred$ of \deffont{internal predicates} and 
\item
$\lset{i}$ for $i{\in}\integer$ of \deffont{internal (level $i$) sets} 
\end{itemize*}
inductively by the rules in Figure~\ref{fig.int.syntax}, where $\kappa$ ranges over finite ordinals.
\item
Define
$$
\begin{array}{r@{\ }l@{\quad\text{and}\quad}r@{\ }l}
\cred
=&
\bigcup_\kappa\cred(\kappa)
&
\lset{i}
=&
\bigcup_\kappa\lset{i}(\kappa) .
\end{array}
$$
\item
Write $\age(X)$ for the least $\kappa$ such that $X{\in}\cred(\kappa)$.
\item
Write $\age(x)$ for the least $\kappa$ such that $x{\in}\lset{i}(\kappa)$.
\end{enumerate*}
\end{defn}

\begin{figure}[t]
$$
\begin{array}{c@{\qquad}c@{\qquad}c}
\begin{prooftree}
a\in\atomsi{i}
\justifies
\myatm{a}\in\lset{i}(\kappa)
\end{prooftree}
&
\begin{prooftree}
\mathcal X\finsubseteq\cred(\kappa)
\justifies
\myand{\mathcal X}\in\cred(\kappa\plus 1)
\end{prooftree}
&
\begin{prooftree}
X{\in}\cred(\kappa)
\justifies
\myneg{X}\in\cred(\kappa\plus 1)
\end{prooftree}
\\[4ex]
\begin{prooftree}
X{\in}\cred(\kappa)\ \ a{\in}\atomsi{i}
\justifies
\myall{a}{X}{\in}\cred(\kappa{+}1)
\end{prooftree}
&
\begin{prooftree}
a{\in}\atomsi{i{+}1}\ \ x{\in}\lset{i}(\kappa)
\justifies
\myelt{x}{a}\in\cred(\kappa{+}1)
\end{prooftree}
&
\begin{prooftree}
X{\in}\cred(\kappa)\ \ a{\in}\atomsi{i\minus 1}
\justifies
\myst{a}{X}{\in}\lset{i}(\kappa\plus 1)
\end{prooftree}
\end{array}
$$
\caption{Syntax of internal predicates and terms}
\label{fig.int.syntax}
\end{figure}

\begin{nttn}
\label{nttn.internal.comprehension}
\begin{itemize}
\item
If $a{\in}\atoms$ we may call $\myatm{a}$ an \deffont{internal atom}. 
\item
If $X{\in}\cred$ we may call $\myst{a}{X}$ an \deffont{internal comprehension}.
\item
We may call $\myatm{a}$ or $\myst{a}{X}$ an \deffont{internal set}.\footnote{So every internal comprehension or internal atom is an internal set.  Another choice of terminology would be to call $\myatm{a}$ an internal atom, $\myst{a}{X}$ an internal set, and $\myatm{a}$ or $\myst{a}{X}$ \emph{internal elements}.

However, note that $\myst{a}{X}$ is not a set and neither is $\myatm{a}$; they are both syntax and we can call them what we like.
}
\end{itemize}
\end{nttn}

\begin{rmrk}
We read through and comment on Definition~\ref{defn.cred}:  
\begin{enumerate*}
\item
$\kappa$ measures the \emph{age} or \emph{stage} of an element; at what point in the induction it is introduced into the datatype.
This is an inductive measure.
\item
If we elide $\kappa$ and levels and simplify, we can rewrite Definition~\ref{defn.cred} semi-formally as follows:
$$
\begin{array}{l@{\ }l}
x\in\lset{} &::=\myatm{a} \mid \myst{a}{X}
\\
X\in\cred &::=\myand{\mathcal X} \mid \myneg{X} \mid \myall{a}{X} \mid \myelt{x}{a}
\end{array}
$$
\item
$\credneg$ represents negation.
$\credand$ represents logical conjunction. 
\item
$\credand$ takes a finite set rather than a pair of terms.
This is a nonessential eccentricity that cuts down on cases later on. 
Truth is represented as $\myand{\varnothing}$.  See Example~\ref{xmpl.credemp}.
\item
$\credall$ represents universal quantification; read $\myall{a}{X}$ as `for all $a$, $X$' or in symbols: `$\forall a.\phi$'.
In $\myall{a}{X}$, $[a]X$ is the nominal atoms-abstraction from Definition~\ref{defn.abs}. 
It implements the binding of the universal quantifier by the standard nominal method.
\item
$\tin$ represents a sets membership; read $\myelt{x}{a}$ as `$x$ is an element of $a$'. 
Note here that $a$ is an atom; it does not literally have any elements. 
$\myelt{x}{a}$ represents the predicate `we believe that $x$ is an element of the variable $a$', or in symbols: `$x\in a$'.
\item
$\myst{a}{X}$ represents sets comprehension; read $\myst{a}{X}$ as `the set of $a$ such that $X$' or in symbols: `$\{a\mid X\}$'.
Again, as standard in nominal techniques, nominal atoms-abstraction is used to represent the binding.

If $a{\in}\atomsi{i}$ then $\myst{a}{X}{\in}\lset{i{+}1}$.
\item
$\myatm{a}$ is a copy of $a{\in}\atoms$ wrapped in some formal syntax $\atm$. 
\end{enumerate*}
\end{rmrk}

\begin{lemm}
\label{lemm.tall.fresh}
Suppose $X{\in}\cred$ and $i{\in}\integer$ and $a,a'{\in}\atomsi{i}$ and $a'\#X$.
Then:
\begin{enumerate*}
\item\label{tall.fresh.alpha.1}
$\myst{a}{X}=\myst{a'}{(a'\ a)\act X}$ and $\myall{a}{X}=\myall{a'}{(a'\ a)\act X}$.
\item\label{tall.fresh.supp}
$a\#\myst{a}{X}$ and $a\#\myall{a}{X}$, and $\supp(\myst{a}{X}),\supp(\myall{a}{X})\subseteq\supp(X){\setminus}\{a\}$.
\item\label{tall.fresh.alpha.2}
For every finite $S\finsubseteq\atomsi{i}$ there exist $b\in\atomsi{i}{\setminus}S$ and $Y{\in}\cred$ such that $\myst{a}{X}=\myst{b}{Y}$ and $\myall{a}{X}=\myall{b}{Y}$.
\end{enumerate*}
\end{lemm}
\begin{proof}
Immediate from Lemma~\ref{lemm.abs.alpha}, and Lemma~\ref{lemm.supp.abs} with Theorem~\ref{thrm.no.increase.of.supp}.
\end{proof}

\begin{lemm}
\label{lemm.age.equivar}
Suppose $i{\in}\integer$ and $a,b\in\atomsi{i}$ and $X\in\cred$ and $x\in\lset{}$.
Then $\age(X)=\age((a\ b)\act X)$ and $\age(x)=\age((a\ b)\act x)$.
\end{lemm}
\begin{proof}
Direct from Theorem~\ref{thrm.equivar}(\ref{equivar.fun}).
\end{proof}

\subsection{Some notation}
\label{subsect.some.useful.notation}

\begin{nttn}
\label{nttn.crediff}
Suppose $X,Y{\in}\cred$ and $\mathcal X\finsubseteq\cred$.
Define syntactic sugar $\myor{\mathcal X}$, $\myimp{X}{Y}$ and $\myiff{X}{Y}$ by 
$$
\begin{array}{r@{\ }l}
\myor{\mathcal X}=&\myneg{(\myand{\{\myneg{X}\mid X\in\mathcal X\}})}
\\
\myimp{X}{Y}=&\myor{\{\myneg{X},Y\}}
\\
\myiff{X}{Y}=&\myand{\{\myimp{X}{Y},\myimp{Y}{X}\}} .
\end{array}
$$ 
\end{nttn}

\begin{xmpl}
\label{xmpl.credemp}
Define $\credfalse\in\cred$ and $\credtrue\in\cred$ by
$$
\credfalse = \myor{\varnothing} 
\quad\text{and}\quad
\credtrue = \myand{\varnothing} .
$$
Intuitively, $\credfalse$ represents the empty disjunction, and $\credtrue$ represents the empty conjunction.
\end{xmpl}

\begin{nttn}
\label{nttn.xata}
Suppose that:
\begin{itemize*}
\item
$i{\in}\integer$ and 
\item
$x=\mathsc{st}(x'){\in}\lset{i}$ is an internal comprehension where $x'\in[\atomsi{i\minus 1}]\cred$ and $a'{\in}\atomsi{i\minus 1}$ and 
\item
$a'\#x$ (equivalently\footnote{By concrete calculations or by Theorem~\ref{thrm.no.increase.of.supp}.} $a'\#x'$).
\end{itemize*}
Then write 
$$
x\at a'
\quad\text{for}\quad x'\at a'.
$$
\end{nttn}

Lemma~\ref{lemm.conc.set} checks that Notation~\ref{nttn.xata} makes sense:
\begin{lemm}
\label{lemm.conc.set}
Suppose $i{\in}\integer$ and $a'{\in}\atomsi{i\minus 1}$.
Suppose $x{\in}\lset{i}$ is an internal comprehension and $a'\#x$ and $X'\in\cred$.
Then:
\begin{enumerate*}
\item
$x\at a'$ is well-defined and $x\at a'\in\cred$.
\item
$x = \myst{a'}{(x\at a')}$.
\item
If $\age(x)=\kappa\plus 1$ then $\age(x\at a')=\kappa$ (meaning that in an inductive argument using age, taking $x\at a'$ strictly decreases the inductive measure).
\item
$(\myst{a'}{X'})\at a' = X'$.
\end{enumerate*}
\end{lemm}
\begin{proof}
\begin{enumerate}
\item
By construction and Lemma~\ref{lemm.abs.exists}.
\item
By construction and Lemma~\ref{lemm.abs.basic}(2).
\item
By construction and Lemma~\ref{lemm.age.equivar}.
\item
By construction and Lemma~\ref{lemm.abs.basic}(1).
\qedhere\end{enumerate}
\end{proof}

Recall $\credfalse{=}\myor{\varnothing}$ from Example~\ref{xmpl.credemp}.
\begin{defn}
\label{defn.credzero}
Suppose $i{\in}\integer$ and $a{\in}\atomsi{i\minus 1}$.  
Define $\credempset^i$ and $\credfullset^i$ by
$$
\begin{array}{r@{\ }l}
\credempset^i =& \myst{a}{\credfalse} = \myst{a}{\myor{\varnothing}}
\\
\credfullset^i =&\myst{a}{\credtrue} = \myst{a}{\myand{\varnothing}} .
\end{array}
$$ 
\end{defn}

We conclude with an easy lemma:
\begin{lemm}
\label{lemm.credzero.at}
Suppose $i{\in}\integer$ and $a{\in}\atomsi{i\minus 1}$.
Then:
\begin{enumerate*}
\item\label{credzero.at.a}
$\credempset^i\at a=\credfalse$ and $\myst{a}{\credfalse}=\credempset^i$, and similarly
$\credfullset^i\at a=\credtrue$ and $\myst{a}{\credtrue}=\credfullset^i$. 
\item\label{credzero.fresh}
$a\#\credfalse$ and $a\#\credtrue$.
\item
Definition~\ref{defn.credzero} does not depend on the choice of $a{\in}\atomsi{i\minus 1}$.
\end{enumerate*}
\end{lemm}
\begin{proof}
\begin{enumerate*}
\item
From Definition~\ref{defn.credzero} and Lemma~\ref{lemm.abs.basic}(1).
\item
From part~1 of this result, since $a\#\credfalse$ and $a\#\credtrue$ by Theorem~\ref{thrm.no.increase.of.supp}.
\item
From part~2 of this result, using Corollary~\ref{corr.stuff}. 
\qedhere\end{enumerate*}
\end{proof}

\section{The sigma-action}
\label{sect.sigma-action}

\subsection{Basic definition and well-definedness}

Intuitively, Definition~\ref{defn.sigma} defines a substitution action.
It is slightly elaborate, especially because of \rulefont{\sigma\credelt a} of Figure~\ref{fig.sub}, so it gets a fancy name (`$\sigma$-action') and we need to make formal and verify that it behaves as a substitution action should; see Remark~\ref{rmrk.is.sigma}.
 
\begin{defn}
\label{defn.sigma}
\label{defn.sub}
Define a \deffont{$\sigma$-action} (sigma-action) to be a family of functions 
$$
\sigma_i:\cred\times\atomsi{i}\times\lset{i}\to \cred
\quad\text{and}\quad
\sigma_{ij}:\lset{j}\times\atomsi{i}\times\lset{i}\to \lset{j}
$$
where $i,j{\in}\integer$, inductively by the rules in Figure~\ref{fig.sub}.
For readability we write
$$
\sigma_i(Z,a,x) \ \text{as}\ Z[a\sm x]
\quad\text{and}\quad
\sigma_i(z,a,x) \ \text{as}\ z[a\sm x]
.
$$
Furthermore in Figure~\ref{fig.sub}:
\begin{itemize*}
\item
In rule \rulefont{\sigma\credand},\ $\mathcal X{\finsubseteq}\cred$.
\item
In rule \rulefont{\sigma\credneg},\ $X{\in}\cred$. 
\item
In rule \rulefont{\sigma\credall},\ $X{\in}\cred$ and $b{\in}\atomsj{j}$ for some $j{\in}\integer$.
\item
In rule \rulefont{\sigma\credelt a},\ $a'{\in}\atomsi{i\minus 1}$.
\item
In rule \rulefont{\sigma\credelt\atm},\ $n$ ranges over \emph{all} atoms in $\atomsi{i}$ (not just those distinct from $a$).
\item
In rule \rulefont{\sigma\credelt b},\ $b{\in}\atomsj{j}$ for some $j{\in}\integer$.
\item
In rule \rulefont{\sigma\credst},\ $X{\in}\cred$ and $c{\in}\atomsi{k}$ for some $k{\in}\integer$.
\end{itemize*}
\end{defn}

\begin{figure}
$$
\begin{array}{l@{\qquad}l@{\ }r@{\ }l@{\quad}l}
\rulefont{\sigma\credand}&
&
(\myand{\mathcal X})[a\sm x]=&\myand{\{X[a\sm x]\mid X{\in}\mathcal X\}}
\\
\rulefont{\sigma\credneg}&
&
(\myneg{X})[a\sm x]=&\myneg{(X[a\sm x])}
\\
\rulefont{\sigma\credall}&
b\#x\limp&
(\myall{b}{X})[a\sm x]=&\myall{b}{(X[a\sm x])}
\\
\rulefont{\sigma\credelt\atm}& 
&(\myelt{y}{a})[a\sm \myatm{n}]=&\myelt{y[a\sm \myatm{n}]}{n}
\\
\rulefont{\sigma\credelt a}& 
&(\myelt{y}{a})[a\sm \myst{a'}{X'}]=&X[a'\sm y[a\sm \myst{a'}{X'}]]
\\
\rulefont{\sigma\credelt b}& 
&(\myelt{y}{b})[a\sm x]=& \myelt{y[a\sm x]}{b}
\\
\rulefont{\sigma a}&
&
\myatm{a}[a\sm x]=&x
\\
\rulefont{\sigma b}&
&
\myatm{b}[a\sm x]=&\myatm{b}
\\
\rulefont{\sigma\credst}&
c\#x\limp&
\myst{c}{X}[a\sm x]=& \myst{c}{(X[a\sm x])}
\end{array}
$$
\caption{The sigma-action (Definition~\ref{defn.sigma})}
\label{fig.sub}
\label{fig.sigma}
\end{figure}

\begin{rmrk}
\label{rmrk.slip.in}
Figure~\ref{fig.sub} slips in no fewer than three abuses of the mathematics:
\begin{enumerate}
\item
We do not know that $X\in\cred$ implies $X[a\sm x]\in\cred$, so we should not write $\myand{\{X[a\sm x]\mid \dots\}}$ on the right-hand side of \rulefont{\sigma\credand}, or indeed $X[a\sm x]$ on the right-hand side of \rulefont{\sigma\credneg}, and so on.

In fact, all right-hand sides of Figure~\ref{fig.sub} are suspect except those of \rulefont{\sigma a} and \rulefont{\sigma b}.
\item
We do not know whether the choice of fresh $a'{\in}\atomsi{i\minus 1}$ in \rulefont{\sigma\credelt a} matters, so we do not know that \rulefont{\sigma\credelt a} is well-defined.
\item
The definition looks inductive at first glance, however in the case of \rulefont{\sigma\credelt a} there is no guarantee that $X$ (on the right-hand side) is smaller than $\myelt{y}{a}$ (on the left-hand side).
The level of $a'$ is strictly lower than the level of $a$, however levels are taken from $\integer$ which is totally ordered but not well-ordered by $\leq$.
\end{enumerate}
In fact:
\begin{itemize}
\item
$X\in\cred$ does indeed imply $X[a\sm x]\in\cred$.
\item
The choice of fresh $a'$ in \rulefont{\sigma\credelt a} is immaterial.
\item
The levels of atoms involved are bounded below (see Definition~\ref{defn.minlev}) 
so we only ever work on a well-founded fragment of $\integer$.
\end{itemize} 
For proofs see Proposition~\ref{prop.X.sigma} and Lemma~\ref{lemm.sigma.alpha}.

Would it be more rigorous to interleave the proofs of these lemmas with the definition, so that at each stage we are confident that what we are writing actually makes sense?
Certainly we could; the reader inclined to worry about this need only read Definition~\ref{defn.sigma} alongside Proposition~\ref{prop.X.sigma} and Lemma~\ref{lemm.sigma.alpha} as a simultaneous inductive argument. 
\end{rmrk}

\begin{rmrk}[Why `minimum level']
Levels are in $\integer$ and are totally ordered by $\leq$ but not well-founded (since integers can `go downwards forever').

However, any (finite) internal predicate or internal set can mention only finitely many levels, so we can calculate the \emph{minimum level} of a predicate or set, which is lower bound on the levels of atoms appearing in that predicate or set.
We will use this lower bound to reason inductively on levels in Propositions~\ref{prop.X.sigma} and~\ref{prop.sigma.sigma}.
\end{rmrk}
 
\begin{defn}
\label{defn.minlev}
Define $\minlev(Z)$ and $\minlev(z)$ the \deffont{minimum level} of $Z$ or $z$, inductively on $Z{\in}\cred$ and $z{\in}\lset{i}$ for $i{\in}\integer$ as follows:
$$
\begin{array}{l@{\ =\ }l}
\minlev(\myatm{a})&\level(a)
\\
\minlev(\myand{\mathcal X})&\f{min}(\{0\}\cup\{\minlev(X)\mid X{\in}\mathcal X\})
\\
\minlev(\myneg{X})&\minlev(X)
\\
\minlev(\myall{a}{X})&\f{min}(\{\level(a),\minlev(X)\})
\\
\minlev(\myelt{x}{a})&\f{min}(\{\minlev(x),\level(a)\})
\\
\minlev(\myst{a}{X})&\f{min}(\{\level(a),\minlev(X)\})
\end{array}
$$
Above, $\f{min}(\mathcal I)$ is the least element of $\mathcal I\finsubseteq\integer$.
We add $0$ in the clause for $\credand$ as a `default value' to exclude calculating a minimum for the empty set; any other fixed integer element would do as well or, if we do not want to make this choice, we can index $\minlev$ over a fixed but arbitrary choice.
The proofs to follow will not care.
\end{defn}

It will be convenient to apply $\minlev$ to a mixed list of internal predicates, atoms, and internal sets:
\begin{nttn}
\begin{itemize}
\item
Define $\minlev(a)=\level(a)$.
\item
If $l=(l_1,l_2,\dots,l_n)$ is a list of elements from $\cred\cup\bigcup_{i{\in}\integer}\lset{i}\cup\atoms$ then we write $\minlev(l)$ for the least element of $\{\minlev(l_1),\dots,\minlev(l_n)\}$.
\end{itemize} 
\end{nttn}

\begin{prop}
\label{prop.X.sigma}
Suppose $i{\in}\integer$ and $a{\in}\atomsi{i}$ and $x{\in}\lset{i}$.
\begin{enumerate}
\item
If $Z{\in}\cred$ then 
\begin{itemize*}
\item
$Z[a\sm x]$ is well-defined, 
\item
$\minlev(Z[a\sm x]){\geq}\minlev(Z,a,x)$, and 
\item
$Z[a\sm x]{\in}\cred$.
\end{itemize*}
\item
If $k{\in}\integer$ and $z{\in}\lset{k}$ then
\begin{itemize*}
\item
$z[a\sm x]$ is well-defined, 
\item
$\minlev(z[a\sm x]){\geq}\minlev(z,a,x)$, and 
\item
$z[a\sm x]{\in}\lset{k}$.
\end{itemize*}
\end{enumerate}
\end{prop}
\begin{proof}
Fix some $k{\in}\integer$.
We prove the Proposition for all $Z,a,x$ and $z,a,x$ with $\minlev(Z,a,x){\geq}k$ and $\minlev(z,a,x){\geq}k$, by induction on 
$$
(\level(a),\age(Z))
\quad\text{and}\quad
(\level(a),\age(z))
$$ 
lexicographically ordered.
Since $k$ was arbitrary, this suffices to prove it for all $Z,a,x$ and $z,a,x$. 

We consider the possibilities for $Z{\in}\cred$:
\begin{itemize*}
\item
\emph{The case of $\myand{\mathcal X}$ for $\mathcal X\finsubseteq\cred$.}\quad

By \figref{fig.sub}{\sigma\credand}
$
Z[a\sm x]{=}\myand{\{X'[a\sm x]{\mid} X'{\in}\mathcal X\}}.
$
We use the inductive hypothesis on each $X'[a\sm x]$ and some easy arithmetic calculations.
\item
\emph{The case of $\myneg{X'}$ for $X'{\in}\cred$.}\quad

By \figref{fig.sub}{\sigma\credneg}
$
Z[a\sm x]=\myneg{(X'[a\sm x])} .
$
We use the inductive hypothesis on $X'[a\sm x]$.
\item
\emph{The case of $\myall{b}{X'}$ for $X'{\in}\cred$ and $b{\in}\atomsj{j}$ for some $j{\in}\integer$.}\quad

Using Lemma~\ref{lemm.tall.fresh}(\ref{tall.fresh.alpha.1}) we may assume without loss of generality that $b\#x$.
By \figref{fig.sub}{\sigma\credall} 
$
(\myall{b}{X'})[a\sm x]=\myall{b'}{(X[a\sm x])}.
$
We use the inductive hypothesis on $X'[a\sm x]$.
\item
\emph{The case of $\myelt{z}{a}$ for $z{\in}\lset{i\minus 1}$.}\quad
There are two sub-cases:
\begin{itemize}
\item
\emph{Suppose $x{=}\myatm{n}$ for some $n{\in}\atomsi{i}$.}\quad
\\
By \figref{fig.sub}{\sigma\credelt\atm} 
$
(\myelt{z}{a})[a\sm x]=\myelt{z[a\sm\myatm{n}]}{n} .
$ 
We use the inductive hypothesis on $z[a\sm\myatm{n}]$.
\item
\emph{Suppose $x{=}\myst{a'}{X'}$ where $X'=x\at a'$ for some fresh $a'{\in}\atomsi{i\minus 1}$ (so $a'\#x,z$).}\quad
\\
By \figref{fig.sub}{\sigma\credelt a} 
$
(\myelt{z}{a})[a\sm x]=X'[a'\sm z[a\sm x]] .
$ 
We have the inductive hypothesis on $z[a\sm x]$.
We also have the inductive hypothesis (since $k{\leq}\level(a'){=}i\minus 1\lneq i{=}\level(a)$)\footnote{$k$ was chosen no greater than the minimum level of $Z$, $a$, and $x$.  Now $x=\myst{a'}{X'}$, and it follows from Definition~\ref{defn.minlev} that $k\leq\level(a')$.}
 on 
$X'[a'\sm z[a\sm x]]$, and this suffices.
\end{itemize}
\item
\emph{The case of $\myelt{z}{c}$ where $c{\in}\atomsi{k}$ and $z{\in}\lset{k\minus 1}$ and $k{\in}\integer$.}\quad

By \figref{fig.sub}{\sigma\credelt b} and the inductive hypothesis.
\end{itemize*}
We consider the possibilities for $z{\in}\lset{k}$:
\begin{itemize}
\item
\emph{The case that $z$ is an internal atom.}\quad

We use \rulefont{\sigma a} or \rulefont{\sigma b} of Figure~\ref{fig.sub}.
\item
\emph{The case that $z$ is an internal comprehension.}\quad

Choose fresh $c{\in}\atomsi{k\minus 1}$ (so $c\#x,z$), so that by Lemma~\ref{lemm.conc.set}(2)
$z=\myst{c}{z\at c}$.
We use the first part of this result and \figref{fig.sub}{\sigma\credst}. 
\qedhere\end{itemize}
\end{proof}

\subsection{Nominal algebraic properties of the sigma-action}

\begin{rmrk}
\label{rmrk.is.sigma}
Several useful properties of the $\sigma$-action from Definition~\ref{defn.sigma}
are naturally expressed as nominal algebra judgements --- equalities subject to freshness conditions \cite{gabbay:nomuae}.
Some are listed for the reader's convenience in Figure~\ref{fig.nomalg}, which goes back to nominal axiomatic studies of substitution from \cite{gabbay:capasn,gabbay:capasn-jv}.

In this paper we are dealing with a concrete model, so the judgements in Figure~\ref{fig.nomalg} are not assumed and are not axioms.
Instead they must be proved; they are propositions and lemmas:
\begin{itemize*}
\item
\rulefont{\sigma\alpha} is Lemma~\ref{lemm.sigma.alpha}.
\item
\rulefont{\sigma\#} is Lemma~\ref{lemm.sigma.hash.always}.
\item
\rulefont{\sigma\sigma} is Proposition~\ref{prop.sigma.sigma}.
\item
\rulefont{\sigma swp} and \rulefont{\sigma asc} are Corollaries~\ref{corr.sigma.swap} and~\ref{corr.sigma.asc}.
\item
\rulefont{\sigma id} is Lemma~\ref{lemm.sigma.id}. 
\item
\rulefont{\sigma ren} is Lemma~\ref{lemm.sigma.ren}. 
\item
\rulefont{\sigma\at} is Lemma~\ref{lemm.sigma.at.b}.
\end{itemize*} 
These are familiar properties of substitution on syntax: for instance 
\begin{itemize*}
\item
\rulefont{\sigma\alpha} looks like an $\alpha$-equivalence property --- and indeed it is ---  and 
\item
\rulefont{\sigma\#} (Lemma~\ref{lemm.sigma.hash.always}) is sometimes called \emph{garbage collection} and corresponds to the property ``if $a$ is not free in $t$ then $t[a\sm s]=t$'', and
\item
\rulefont{\sigma\sigma} (Proposition~\ref{prop.sigma.sigma}) is often called the \emph{substitution lemma}.
See the discussion in Remark~\ref{rmrk.substitution.lemma}.
\end{itemize*}
But, the proofs of these properties that we see in this paper are not replays of the familiar syntactic properties.

This is because the $\sigma$-action on $\cred$ is not a simple `tree-grafting' operation ---  not even a capture-avoiding one --- because of \rulefont{\sigma\credelt a} in Figure~\ref{fig.sub}.
The proofs work, but we cannot take this for granted, and they require checking.
\end{rmrk}

\begin{figure}
$$
\begin{array}{l@{\qquad}r@{\ }r@{\ }l}
\rulefont{\sigma\alpha}&
b'\#Z\limp&Z[b\sm y]=&((b'\ b)\act Z)[b'\sm y]
\\
\rulefont{\sigma\#}& 
b\#Z\limp& Z[b\sm y]=&Z 
\\
\rulefont{\sigma\sigma}&
a\#y\limp&Z[a\sm x][b\sm y]=&Z[b\sm y][a\sm x[b\sm y]]
\\
\rulefont{\sigma swp}&
a\#y,\ b\#x\limp&Z[a\sm x][b\sm y]=&Z[b\sm y][a\sm x]
\\
\rulefont{\sigma asc}&
a\#y,\ b\#Z\limp&Z[a\sm x[b\sm y]]=&Z[a\sm x][b\sm y]
\\
\rulefont{\sigma id}&
& Z[a\sm \myatm{a}]=&Z  
\\
\rulefont{\sigma ren}&
a'\#Z\limp&Z[a\sm \myatm{a'}]=&(a'\ a)\act Z
\\
\rulefont{\sigma\at}&
c\#x\limp&(z\at c)[a\sm x]=&z[a\sm x]\at c
\end{array}
$$
\caption{Further nominal algebra properties of the $\sigma$-action}
\label{fig.nomalg}
\end{figure}

\subsubsection{Alpha-equivalence of the sigma-action}

\begin{lemm}[\rulefont{\sigma\alpha}]
\label{lemm.sigma.alpha}
Suppose $i{\in}\integer$ and $a,a'{\in}\atomsi{i}$ and $x{\in}\lset{i}$.
Suppose $Z{\in}\cred$ and $a'\#Z$ and $z{\in}\lset{k}$ and $a'\#z$.
Then:
\begin{enumerate*}
\item
$Z[a\sm x]=((a'\ a)\act Z)[a'\sm x]$ and $z[a\sm x]=((a'\ a)\act z)[a'\sm x]$.
\item\label{alpha.supp}
$\supp(Z[a\sm x])\subseteq(\supp(Z){\setminus}\{a\})\cup\supp(x)$ and
$\supp(z[a\sm x])\subseteq(\supp(z){\setminus}\{a\})\cup\supp(x)$.
\item\label{alpha.fresh}
If $a\#x$ then $a\#Z[a\sm x]$ and $a\#z[a\sm x]$.
\end{enumerate*}
\end{lemm}
\begin{proof}
By induction on 
$$
\age(Z)
\quad\text{and}\quad \age(z).
$$
We consider the possibilities for $Z{\in}\cred$:
\begin{itemize}
\item
\emph{The case of $\myand{\mathcal X}$ for $\mathcal X\finsubseteq\cred$.}\quad
By Lemma~\ref{lemm.finpow.support} $a\#X'$ for every $X'{\in}\mathcal X$, so by the inductive hypothesis $X'[a\sm x]=((a'\ a)\act X')[a'\sm x]$.
We use \figref{fig.sub}{\sigma\credand} and Theorem~\ref{thrm.equivar}.
\item
\emph{The case of $\myneg{X}$ for $X{\in}\cred$.}\quad
By \figref{fig.sub}{\sigma\credneg} and the inductive hypothesis for $X$.
\item
\emph{The case of $\myall{b}{X}$ for $X{\in}\cred$ and $b{\in}\atomsj{j}$ for some $j{\in}\integer$.}\quad
Using Lemma~\ref{lemm.tall.fresh}(\ref{tall.fresh.alpha.1}) we may assume without loss of generality that $b\#x$.
We use \figref{fig.sub}{\sigma\credall} and the inductive hypothesis.
\item
\emph{The case of $\myelt{y}{a}$ for some $y{\in}\lset{i\minus 1}$.}\quad
There are two sub-cases:
\begin{itemize}
\item
\emph{Suppose $x{=}\myatm{n}$ for some $n{\in}\atomsi{i}$.}\quad
We reason as follows:
$$
\begin{array}{r@{\ }l@{\quad}l}
(\myelt{y}{a})[a\sm \myatm{n}]
=&
\myelt{y[a\sm \myatm{n}]}{n}
&\figref{fig.sub}{\sigma\credelt\atm}
\\
=&
\myelt{((a'\ a)\act y)[a'\sm \myatm{n}]}{n}
&\text{Ind hyp for $y$}
\\
=&
(\myelt{((a'\ a)\act y)}{a'})[a'\sm \myatm{n}]
&\figref{fig.sub}{\sigma\credelt\atm}
\\
=&
((a'\ a)\act (\myelt{y}{a}))[a'\sm \myatm{n}]
&\text{Theorem~\ref{thrm.equivar}}
\end{array}
$$
\item
\emph{Suppose $x{=}\myst{b'}{X'}$ where $X'=x\at b'$ for some fresh $b'{\in}\atomsi{i\minus 1}$ (so $b'\#x,y,z$).}\quad
\\
We reason as follows:
$$
\begin{array}{@{\hspace{-2em}}r@{\ }l@{\quad}l}
(\myelt{y}{a})[a\sm x]
=&
X'[b'\sm y[a\sm x]]
&\figref{fig.sub}{\sigma\credelt a}
\\
=&
X'[b'\sm ((a'\ a)\act y)[a'\sm x]]
&\text{Ind hyp for $y$}
\\
=&
(\myelt{(a'\ a)\act y}{a'})[a'\sm x]
&\figref{fig.sub}{\sigma\credelt a}
\end{array}
$$
\end{itemize}
\item
\emph{The case $\myelt{y}{b}$ for $j{\in}\integer$ and $b{\in}\atomsj{j}$ and $y{\in}\lset{j\minus 1}$.}\quad
We reason as follows:
$$
\begin{array}{r@{\ }l@{\quad}l}
(\myelt{y}{b})[a\sm x]
=&
\myelt{y[a\sm x]}{b}
&\figref{fig.sub}{\sigma\credelt b}
\\
=&
\myelt{((a'\ a)\act y)[a'\sm x]}{b}
&\text{Ind hyp for $y$}
\\
=&
(\myelt{((a'\ a)\act y)}{b})[a'\sm x]
&\figref{fig.sub}{\sigma\credelt b}
\\
=&
((a'\ a)\act (\myelt{y}{b}))[a'\sm x]
&\text{Theorem~\ref{thrm.equivar}}
\end{array}
$$
\end{itemize}
We consider the possibilities for $z{\in}\lset{k}$:
\begin{itemize}
\item
\emph{The case that $z$ is an internal atom.}\quad
We use \rulefont{\sigma a} or \rulefont{\sigma b} of Figure~\ref{fig.sub}.
\item
\emph{The case that $z$ is an internal comprehension.}\quad
We use 
Lemma~\ref{lemm.conc.set}(2\&3)
for a fresh $c{\in}\atomsi{k\minus 1}$ (so $c\#z$), \rulefont{\sigma\credst}, and the inductive hypothesis.
\end{itemize}

For part~2, we note that by Theorem~\ref{thrm.no.increase.of.supp} and Proposition~\ref{prop.pi.supp} 
$$
\begin{array}{r@{\ }l}
\supp(Z[a\sm x])\subseteq&\supp(Z)\cup\{a\}\cup\supp(x)
\quad\text{and}
\\
\supp(((a'\ a)\act Z)[a'\sm x])\subseteq&(a'\ a)\act \supp(Z)\cup\{a'\}\cup\supp(x) .
\end{array}
$$
We take a sets intersection.
The case of $z$ is similar.

Part~3 follows, recalling from Notation~\ref{nttn.fresh} that $a\#x$ means $a{\not\in}\supp(x)$. 
\end{proof}

\begin{rmrk}
\emph{Note for experts:}\ We could set Definition~\ref{defn.sigma} up differently: we could have $\sigma_i$ and $\sigma_{ij}$ input abstractions 
$$
\sigma_i:([\atomsi{i}]\cred)\times\lset{i}\to \cred
\quad\text{and}\quad
\sigma_{ij}:([\atomsi{i}]\lset{j})\times\lset{i}\to \lset{j}
.
$$
Then Lemma~\ref{lemm.sigma.alpha} would become immediate from Lemmas~\ref{lemm.supp.abs} and~\ref{lemm.abs.alpha} and Theorem~\ref{thrm.no.increase.of.supp}.
This is nice but note that we incur a well-definedness proof-obligation that the choice of name for the abstracted atom does not matter.
There is probably still a net gain but it is not quite as great as it might first seem.
For this reason, we use the more elementary set-up in Definition~\ref{defn.sigma}.
\end{rmrk}

\subsubsection{Property \rulefont{\sigma\#} (garbage collection)}

\begin{lemm}[\rulefont{\sigma\#}]
\label{lemm.sigma.hash.always}
Suppose $i{\in}\integer$ and $a{\in}\atomsi{i}$ and $x{\in}\lset{i}$ and $Z{\in}\cred$ and $z{\in}\lset{k}$ for $k{\in}\integer$.
Then 
$$
\begin{array}{l@{\ }r@{\ }l}
a\#Z\limp& Z[a\sm x]=&Z 
\\
a\#z\limp& z[a\sm x]=&z.
\end{array}
$$
\end{lemm}
\begin{proof}
By induction on 
$$
\age(Z)\quad\text{and}\quad \age(z).
$$ 
We consider the possibilities for $Z{\in}\cred$:  
\begin{itemize*}
\item
\emph{The case of $\myand{\mathcal X}$ for $\mathcal X\finsubseteq\cred$.}\quad

By \figref{fig.sub}{\sigma\credand} $(\myand{\mathcal X})[a\sm x]=\myand{\{X[a\sm x]\mid X{\in}\mathcal X\}}$.
By Lemma~\ref{lemm.finpow.support}(\ref{item.finpow.support.subset}) $a\#X$ for every $X{\in}\mathcal X$.
We use the inductive hypothesis on each $X$.
\item
\emph{The case of $\myneg{X}$ for $X{\in}\cred$.}\quad

By \figref{fig.sub}{\sigma\credneg} $\myneg{X}[a\sm x]=\myneg{X[a\sm x]}$.
We use the inductive hypothesis on $X$.
\item
\emph{The case of $\myall{b}{X}$ for $X{\in}\cred$ and $b{\in}\atomsj{j}$ for some $j{\in}\integer$.}

Using Lemma~\ref{lemm.tall.fresh}(\ref{tall.fresh.alpha.1}) we may assume without loss of generality that $b\#x$.
By \figref{fig.sub}{\sigma\credall} $(\myall{b}{X})[a\sm x]=\myall{b}{(X[a\sm x])}$.
We use the inductive hypothesis on $X$.
\item
\emph{The case of $\myelt{y}{a}$ for $i{\in}\integer$ and $y{\in}\lset{i\minus 1}$.}\quad

This is impossible because we assumed $a\#Z$.
\item
\emph{The case of $\myelt{y}{b}$ for $j{\in}\integer$ and $b{\in}\atomsj{j}$ and $y{\in}\lset{j\minus 1}$.}\quad

By \figref{fig.sub}{\sigma\credelt b} $(\myelt{y}{b})[a\sm x]=\myelt{y[a\sm x]}{b}$.
We use the inductive hypothesis on $y$.
\end{itemize*}
We consider the possibilities for $z{\in}\lset{k}$:
\begin{itemize}
\item
If $z$ is an internal atom then we reason using \rulefont{\sigma a} or \rulefont{\sigma b} of Figure~\ref{fig.sub}.
\item
If $z$ is an internal comprehension then we use Lemma~\ref{lemm.conc.set}(2\&3)
for a fresh $c{\in}\atomsi{k\minus 1}$ (so $c\#z$), \rulefont{\sigma\credst}, and the inductive hypothesis.
\qedhere\end{itemize}
\end{proof}

Recall $\credfalse{=}\myor{\varnothing}$ and $\credtrue{=}\myand{\varnothing}$ from Example~\ref{xmpl.credemp}.
Corollary~\ref{corr.credemp.sigma} is an easy consequence of Lemma~\ref{lemm.sigma.hash.always} and will be useful later:
\begin{corr}
\label{corr.credemp.sigma} 
Suppose $i{\in}\integer$ and $a{\in}\atomsi{i}$ and $x{\in}\lset{i}$.
Then
$$
\credfalse[a\sm x]=\credfalse 
\quad\text{and}\quad
\credtrue[a\sm x]=\credtrue
.
$$
\end{corr}
\begin{proof}
By Theorem~\ref{thrm.no.increase.of.supp} $\supp(\credfalse){=}\varnothing$ so that $a\#x$.
We use Lemma~\ref{lemm.sigma.hash.always}.
Similarly for $\credtrue$.
\end{proof}

\subsubsection{$\sigma$ commutes with atoms-concretion}

Lemma~\ref{lemm.sigma.at.b} will be useful later, starting with Proposition~\ref{prop.sigma.sigma}:
\begin{lemm}[\rulefont{\sigma\at}]
\label{lemm.sigma.at.b}
Suppose $i{\in}\integer$ and $a{\in}\atomsi{i}$ and $x{\in}\lset{i}$.
Suppose $k{\in}\integer$ and $z\in\lset{k}$ is an internal comprehension and $c{\in}\atomsi{k\minus 1}$ and $c\#z,x$.
Then 
$$
(z\at c)[a\sm x]=z[a\sm x]\at c.
$$
($z\at c$ is from Notation~\ref{nttn.xata}.) 
\end{lemm}
\begin{proof}
By Lemma~\ref{lemm.conc.set}(2) we may write $z=\myst{c}{Z}$ where $Z=z\at c$.
We reason as follows:
\begin{align*}
\begin{array}[b]{r@{\ }l@{\qquad}l}
(z\at c)[a\sm x]
=&
Z[a\sm x]
&Z=z\at c
\\
=&
\myst{c}{(Z[a\sm x])}\at c 
&\text{Lemma~\ref{lemm.conc.set}(2)}
\\
=&
(\myst{c}{Z}[a\sm x])\at c 
&\figref{fig.sub}{\sigma\credst}, c\#x
\\
=&
(z[a\sm x])\at c 
&z=\myst{c}{Z}                    \tag*{\qEd}
\end{array}
\end{align*}
\def\popQED{}
\end{proof}

\subsubsection{$\sigma$ commutes with itself: the `substitution lemma'}

\begin{prop}
\label{prop.sigma.sigma}
Suppose $Z{\in}\cred$ and $k{\in}\integer$ and $z{\in}\lset{k}$. 
Suppose $i{\in}\integer$ and $a{\in}\atomsi{i}$ and $x{\in}\lset{i}$ and suppose $j{\in}\integer$ and $b{\in}\atomsj{j}$ and $y{\in}\lset{j}$ and $a\#y$.
Then 
$$
\begin{array}{r@{\ }l}
Z[a\sm x][b\sm y] =& Z[b\sm y][a\sm x[b\sm y]]
\\
z[a\sm x][b\sm y] =& z[b\sm y][a\sm x[b\sm y]] .
\end{array}
$$
\end{prop}
\begin{proof}
For brevity we may write 
$$
\sigma\text{ for }[a\sm x][b\sm y]\quad\text{and}\quad 
\sigma'\text{ for }[b\sm y][a\sm x[b\sm y]].
$$
Fix some $k{\in}\integer$. 
We prove the Lemma for all $Z,a,x,b,y$ and $z,a,x,b,y$ with $\minlev(Z,a,x,b,y){\geq}k$ and $\minlev(z,a,x,b,y){\geq}k$ (Definition~\ref{defn.minlev}), reasoning by induction on 
$$
(\level(a){+}\level(b),\age(Z))
\quad\text{and}\quad 
(\level(a){+}\level(b),\age(z))
$$ 
lexicographically ordered.
Since $k$ was arbitrary, this suffices to prove it for all $Z,a,x,b,y$ and $z,a,x,b,y$. 

We consider the possibilities for $Z{\in}\cred$:
\begin{itemize}
\item
\emph{The case of $\myand{\mathcal X}$ for $\mathcal X\finsubseteq\cred$.}
We use rule \rulefont{\sigma\credand} of Figure~\ref{fig.sub} and the inductive hypothesis.
\item
\emph{The case of $\myneg{X}$ for $X{\in}\cred$.}
We use \rulefont{\sigma\credneg} of Figure~\ref{fig.sub} and the inductive hypothesis.
\item
\emph{The case of $\myall{a'}{X}$ for $X{\in}\cred$ and $a'{\in}\atomsi{i'}$ for some $i'{\in}\integer$.}
We use Lemma~\ref{lemm.tall.fresh}(\ref{tall.fresh.alpha.1}) to assume without loss of generality that $a'\#x,y$, and then we use \rulefont{\sigma\credall} of Figure~\ref{fig.sub} and the inductive hypothesis.
\item
\emph{The case of $\myelt{z}{b}$ for $z{\in}\lset{j\minus 1}$ where $j{\in}\integer$.}\quad
There are two sub-cases:
\begin{itemize}
\item
\emph{Suppose $y{=}\myatm{n}$ for some $n{\in}\atomsi{i}$ other than $a$} (we assumed $a\#y$ so $n{=}a$ is impossible).\quad

We reason as follows:
$$
\begin{array}{@{\hspace{-2em}}r@{\ }l@{\quad}l}
(\myelt{z}{b})&[a\sm x][b\sm \myatm{n}]
\\
=&
(\myelt{z[a\sm x]}{b})[b\sm \myatm{n}]
&\figref{fig.sub}{\sigma\credelt b}
\\
=&
\myelt{z\sigma}{n}
&\figref{fig.sub}{\sigma\credelt\atm}
\\
=&
\myelt{z\sigma'}{n}
&\text{IH }\age(z){<}\age(\myelt{z}{b}),\ a\#y
\\
=&
(\myelt{z[b\sm\myatm{n}]}{n})[a\sm x[b\sm \myatm{n}]]
&\figref{fig.sub}{\sigma\credelt b}
\\
=&
(\myelt{z}{b})[b\sm \myatm{n}][a\sm x[b\sm \myatm{n}]]
&\figref{fig.sub}{\sigma\credelt\atm}
\end{array}
$$
\item
\emph{Suppose $y{=}\myst{b'}{Y'}$ where $Y'=y\at b'$ for some fresh $b'{\in}\atomsj{j\minus 1}$ (so $b'\#z,x,y$ and $k{\leq}\level(b')$).}\quad

Note by Theorem~\ref{thrm.no.increase.of.supp} that $a\#Y'$ and $b'\#x[b\sm y]$.
We reason as follows:
$$
\begin{array}{@{\hspace{-2em}}r@{\ }l@{\quad}l}
(\myelt{z}{b})&[a\sm x][b\sm y]
\\
=&
(\myelt{z[a\sm x]}{b})[b\sm y]
&\figref{fig.sub}{\sigma\credelt b}
\\
=&
Y'[b'\sm z\sigma]
&\figref{fig.sub}{\sigma\credelt a}
\\
=&
Y'[b'\sm z\sigma']
&\text{IH }\age(z){<}\age(\myelt{z}{b}),\ a\#y
\\
=&
Y'[a\sm x[b\sm y]][b'\sm z\sigma']
&\text{Lemma~\ref{lemm.sigma.hash.always}},\ a\#Y'
\\
=&
Y'[b'\sm z[b\sm y]]\,[a\sm x[b\sm y]]
&\text{IH }\level(b'){<}\level(b),\ b'\#x[b\sm y]
\\
=&
(\myelt{z}{b})[b\sm y][a\sm x[b\sm y]]
&\figref{fig.sub}{\sigma\credelt a}
\end{array}
$$
\end{itemize}
\item
\emph{The case of $\myelt{z}{a}$ for $z{\in}\lset{i\minus 1}$ where $i{\in}\integer$.}\quad
There are two sub-cases:
\begin{itemize}
\item
\emph{Suppose $x{=}\myatm{n}$ for some $n{\in}\atomsi{i}$.}\quad

If $n{\neq}b$ then we reason as follows:
$$
\begin{array}{@{\hspace{-2em}}r@{\ }l@{\quad}l}
(\myelt{z}{a})&[a\sm \myatm{n}][b\sm y]
\\
=&
(\myelt{z[a\sm\myatm{n}]}{n})[b\sm y]
&\figref{fig.sub}{\sigma\credelt\atm}
\\
=&
\myelt{z\sigma}{n}
&\figref{fig.sub}{\sigma\credelt b}
\\
=&
\myelt{z\sigma'}{n}
&\text{IH }\age(z){<}\age(\myelt{z}{a}),\ a\#y
\\
=&
\myelt{z[b\sm y][a\sm\myatm{n}]}{n}
&\rulefont{\sigma b}\ n{\neq} b
\\
=&
(\myelt{z[b\sm y]}{a})[a\sm\myatm{n}]
&\figref{fig.sub}{\sigma\credelt\atm}
\\
=&
(\myelt{z[b\sm y]}{a})[a\sm\myatm{n}[b\sm y]]
&\rulefont{\sigma b}\ n{\neq}b
\\
=&
(\myelt{z}{a})[b\sm y][a\sm\myatm{n}[b\sm y]]
&\figref{fig.sub}{\sigma\credelt b}
\end{array}
$$
If $n{=}b$ so that $x{=}\myatm{b}$, and $y{=}\myatm{m}$ for some $m{\in}\atomsj{j}$ other than $a$, then we reason as follows:
$$
\begin{array}{@{\hspace{-2em}}r@{\ }l@{\quad}l}
(\myelt{z}{a})&[a\sm \myatm{b}][b\sm \myatm{m}]
\\
=&
(\myelt{z[a\sm\myatm{b}]}{b})[b\sm \myatm{m}]
&\figref{fig.sub}{\sigma\credelt\atm}
\\
=&
\myelt{z[a\sm\myatm{b}][b\sm\myatm{m}]}{m}
&\figref{fig.sub}{\sigma\credelt\atm}
\\
=&
\myelt{z[b\sm\myatm{m}][a\sm\myatm{b}[b\sm\myatm{m}]]}{m}
&\text{IH }\age(z){<}\age(\myelt{z}{a}),\ a\#m
\\
=&
(\myelt{z[b\sm\myatm{m}]}{m})[a\sm\myatm{b}[b\sm\myatm{m}]]
&\figref{fig.sub}{\sigma\credelt b}
\\
=&
(\myelt{z}{b})[b\sm\myatm{m}][a\sm\myatm{b}[b\sm\myatm{m}]]
&\figref{fig.sub}{\sigma\credelt\atm}
\end{array}
$$
If $n{=}b$ so that $x{=}\myatm{b}$, and $y{=}\myst{b'}{Y'}$ where $Y'=y\at b'$ for some fresh $b'{\in}\atomsj{j\minus 1}$ (so $b'\#z,n,y$ and $k{\leq}\level(b')$) then we reason as follows (note by Theorem~\ref{thrm.no.increase.of.supp} that $a\#Y'$ and $b'\#x[b\sm y]$):
$$
\begin{array}{@{\hspace{-2em}}r@{\ }l@{\quad}l}
(\myelt{z}{a})&[a\sm \myatm{b}][b\sm y]
\\
=&
(\myelt{z[a\sm\myatm{b}]}{b})[b\sm y]
&\figref{fig.sub}{\sigma\credelt\atm}
\\
=&
Y'[b'\sm z\sigma]
&\figref{fig.sub}{\sigma\credelt a}
\\
=&
Y'[b'\sm z\sigma']
&\text{IH }\age(z){<}\age(\myelt{z}{a}),\ a\#y
\\
=&
Y'[a\sm\myatm{b}[b\sm y]][b'\sm z\sigma']
&\text{Lemma~\ref{lemm.sigma.hash.always}},\ a\#Y'
\\
=&
Y'[b'\sm z[b\sm y]][a\sm\myatm{b}[b\sm y]]
&\text{IH }\level(b'){<}\level(b),\ b'\#x[b\sm y]
\\
=&
(\myelt{z}{a})[b\sm y][a\sm\myatm{b}[b\sm y]]
&\figref{fig.sub}{\sigma\credelt b}
\end{array}
$$
\item
\emph{Suppose $x{=}\myst{a'}{X'}$ where $X'=x\at a'$ for some fresh $a'{\in}\atomsi{i\minus 1}$ (so $a'\#z,x,y$ and $k{\leq}\level(a')$).}\quad

We reason as follows:
$$
\begin{array}{@{\hspace{-2em}}r@{\ }l@{\quad}l}
(\myelt{z}{a})[a\sm x][b\sm y]
=&
X'[a'\sm z[a\sm x]]\,[b\sm y]
&\figref{fig.sub}{\sigma\credelt a}
\\
=&
X'[b\sm y][a'\sm z\sigma]
&\text{IH }k{\leq}\level(a'){=}\level(a)\minus 1,\ a'\#y
\\
=&
X'[b\sm y][a'\sm z\sigma']
&\text{IH }\age(z){<}\age(\myelt{z}{a}),\ a\#y
\\
=&
(x[b\sm y]\at a')[a'\sm z\sigma']
&\text{Lemma~\ref{lemm.sigma.at.b}},\ a'\#y
\\
=&
(\myelt{z[b\sm y]}{a})[a\sm x[b\sm y]]
&\figref{fig.sub}{\sigma\credelt a}
\\
=&
(\myelt{z}{a})[b\sm y][a\sm x[b\sm y]]
&\figref{fig.sub}{\sigma\credelt a},\ a\#y
\end{array}
$$
\end{itemize}
\item
\emph{The case of $\myelt{z}{c}$ for $k{\in}\integer$ and $c{\in}\atomsi{k}$ and $z{\in}\lset{k\minus 1}$.}\quad
We reason as follows:
$$
\begin{array}{@{\hspace{-2em}}r@{\ }l@{\quad}l}
(\myelt{z}{c})[a\sm x][b\sm y]
=&
\myelt{z\sigma}{c}
&\text{\figref{fig.sub}{\sigma\credelt b}, twice}
\\
=&
\myelt{z\sigma'}{c}
&\text{IH }\age(z){<}\age(\myelt{z}{a}),\ a\#y
\\
=&
(\myelt{z}{c})[b\sm y][a\sm x[b\sm y]]
&\text{\figref{fig.sub}{\sigma\credelt b}, twice}
\end{array}
$$
\end{itemize}
We consider the possibilities for $z{\in}\lset{k}$:
\begin{itemize}
\item
If $z$ is an internal atom then we reason using \rulefont{\sigma a} and \rulefont{\sigma b} of Figure~\ref{fig.sub}.
\item
If $z$ is an internal comprehension then we use 
Lemma~\ref{lemm.conc.set}(2\&3)
for a fresh $c{\in}\atomsi{k\minus 1}$ (so $c\#z$), \rulefont{\sigma\credst}, and the inductive hypothesis.
\qedhere\end{itemize}
\end{proof}

\begin{rmrk}
\label{rmrk.substitution.lemma}
Were Proposition~\ref{prop.sigma.sigma} about the syntax of first-order logic or the $\lambda$-calculus, then it could be called \emph{the substitution lemma}, and the proof would be a routine induction on syntax.

In fact, even in the case of first-order logic or the $\lambda$-calculus, the proof is not routine.  Issues with binders (\figref{fig.sub}{\sigma\credelt a}, and one explicit in \rulefont{\sigma\credst}) were the original motivation for my thesis \cite{gabbay:thesis} and for nominal techniques in general. 

For a standard non-rigorous non-nominal proof of the substitution lemma see \cite{barendregt:lamcss}; for a detailed discussion of the lemma in the context of Nominal Isabelle, see \cite{urban:sub-discussion} which includes many further references.

But the proof of Proposition~\ref{prop.sigma.sigma} is not just a replay of the proofs; neither in the `classic' sense of \cite{barendregt:lamcss} nor in the `nominal' sense of \cite{gabbay:thesis,urban:sub-discussion}.
This is because of the interaction of $\credelt$ with the $\sigma$-action, mostly because of \rulefont{\sigma\credelt a} (to a lesser extent also because of the nominal binder \rulefont{\sigma\credst}).
\end{rmrk}

\begin{corr}[\rulefont{\sigma swp}]
\label{corr.sigma.swap}
Suppose $Z{\in}\cred$ and $k{\in}\integer$ and $z{\in}\lset{k}$. 
Suppose $i{\in}\integer$ and $a{\in}\atomsi{i}$ and $x{\in}\lset{i}$ and suppose $j{\in}\integer$ and $b{\in}\atomsj{j}$ and $y{\in}\lset{j}$.
Suppose $a\#y$ and $b\#x$.
Then 
$$
\begin{array}{r@{\ }l}
Z[a\sm x][b\sm y] =& Z[b\sm y][a\sm x]
\\
z[a\sm x][b\sm y] =& z[b\sm y][a\sm x] .
\end{array}
$$
\end{corr}
\begin{proof}
From Proposition~\ref{prop.sigma.sigma} and Lemma~\ref{lemm.sigma.hash.always}. 
\end{proof}

\begin{corr}[\rulefont{\sigma asc}]
\label{corr.sigma.asc}
Suppose $Z{\in}\cred$ and $k{\in}\integer$ and $z{\in}\lset{k}$. 
Suppose $i{\in}\integer$ and $a{\in}\atomsi{i}$ and $x{\in}\lset{i}$ and suppose $j{\in}\integer$ and $b{\in}\atomsj{j}$ and $y{\in}\lset{j}$.
Suppose $a\#y$ and $b\#Z,z$.\footnote{We expect a stronger version of Corollary~\ref{corr.sigma.asc} to be possible in which we do not assume $a\#y$.
However, the proof would require an induction resembling the proof of Proposition~\ref{prop.sigma.sigma} ---  the proof assuming $a\#y$ can piggyback on the induction already given in Proposition~\ref{prop.sigma.sigma}.
We will not need this stronger version, so we do not bother.}
Then 
$$
\begin{array}{r@{\ }l}
Z[a\sm x[b\sm y]] =& Z[a\sm x][b\sm y]
\\
z[a\sm x[b\sm y]] =& z[a\sm x][b\sm y] .
\end{array}
$$
\end{corr}
\begin{proof}
From Proposition~\ref{prop.sigma.sigma} and Lemma~\ref{lemm.sigma.hash.always}. 
\end{proof}

\subsubsection{\rulefont{\sigma id}: substitution for atoms and its corollaries}

We called $\myatm{a}$ in Definition~\ref{defn.lsets} an \emph{internal atom}. 
Atoms in nominal techniques interpet variables, so if we call $\myatm{a}$ an internal atom this should suggest that $\myatm{a}$ should behave like a variable (or a variable symbol).
Rules \rulefont{\sigma a} and \rulefont{\sigma b} from Figure~\ref{fig.sub} are consistent with that, and Lemma~\ref{lemm.sigma.id} makes formal more of this intuition:
\begin{lemm}[\rulefont{\sigma id}]
\label{lemm.sigma.id}
Suppose $i{\in}\integer$ and $a{\in}\atomsi{i}$. 
Then:
\begin{enumerate*}
\item
If $Z{\in}\cred$ then $Z[a\sm\myatm{a}]=Z$.
\item
If $k{\in}\integer$ and $z{\in}\lset{k}$ then $z[a\sm\myatm{a}]=z$.
\end{enumerate*}
\end{lemm}
\begin{proof}
We reason by induction on 
$$
\age(Z)
\quad\text{and}\quad \age(z).
$$
We consider the possibilities for $Z{\in}\cred$:
\begin{itemize}
\item
If $Z=\myand{\mathcal Z}$ for $\mathcal Z\finsubseteq\cred$ or $Z=\myneg{Z'}$ for $Z'{\in}\cred$ then we use rules \rulefont{\sigma\credand} and \rulefont{\sigma\credneg} of Figure~\ref{fig.sub} and the inductive hypothesis. 
\item
If $Z=\myall{a'}{Z'}$ for $Z'{\in}\cred$ and $a'{\in}\atomsi{i'}$ for some $i'{\in}\integer$ then we use \rulefont{\sigma\credall} of Figure~\ref{fig.sub} and the inductive hypothesis.
\item
If $Z=(\myelt{z}{b})$ for $j{\in}\integer$ and $b{\in}\atomsj{j}$ and $z{\in}\lset{j\minus 1}$ then we use rule \rulefont{\sigma\credelt b} of Figure~\ref{fig.sub} and the inductive hypothesis.
\item
If $Z=(\myelt{z}{a})$ for $z{\in}\lset{i\minus 1}$ then we use \rulefont{\sigma\credelt\atm} of Figure~\ref{fig.sub} and the inductive hypothesis for $z$.
\end{itemize}
We consider the possibilities for $z{\in}\lset{k}$:
\begin{itemize}
\item
If $z$ is an atom then we reason using \rulefont{\sigma a} or \rulefont{\sigma b} of Figure~\ref{fig.sub}.
\item
If $z$ is an internal comprehension then we use Lemma~\ref{lemm.conc.set}(2\&3)
for a fresh $c{\in}\atomsi{k\minus 1}$ (so $c\#z$), \rulefont{\sigma\credst}, and the inductive hypothesis.
\qedhere\end{itemize}
\end{proof}

Given what we have so far, Lemma~\ref{lemm.sigma.ren} is not hard to prove.
\begin{lemm}[\rulefont{\sigma ren}]
\label{lemm.sigma.ren}
Suppose $i{\in}\integer$ and $a,a'{\in}\atomsi{i}$. 
Then:
\begin{itemize*}
\item
If $Z{\in}\cred$ and $a'\#Z$ then $Z[a\sm\myatm{a'}]=(a'\ a)\act Z$.
\item
If $k{\in}\integer$ and $z{\in}\lset{k}$ and $a'\#z$ then $z[a\sm\myatm{a'}]=(a'\ a)\act z$.
\end{itemize*}
\end{lemm}
\begin{proof}
Suppose $Z{\in}\cred$ and $a'\#Z$.
We note by Lemma~\ref{lemm.sigma.alpha}(1) (since $a'\#Z$) that $Z[a\sm\myatm{a'}]=((a'\ a)\act Z)[a'\sm\myatm{a'}]$ and use Lemma~\ref{lemm.sigma.id}(1).
The case of $z{\in}\lset{k}$ is exactly similar.
\end{proof}

\subsection{Sigma-algebras and \theory{SUB}}

We can now observe that our sigma-action is consistent with the nominal algebra literature in the following sense:
\begin{thrm}
\label{thrm.capasn}
The syntaxes of internal predicates and internal terms, with the sigma-action from Definition~\ref{defn.sigma}, are \emph{sigma-algebras} in the sense of \cite{gabbay:semooc,gabbay:repdul}, and models of \theory{SUB} in the sense of \cite{gabbay:capasn-jv}.

Concretely, this means that the sigma-action from Definition~\ref{defn.sigma} and Figure~\ref{fig.sigma} should
\begin{itemize*}
\item
distribute through $\credand$, $\credneg$, and 
\item
distribute in a capture-avoiding manner through $\credall$, and $\credst$, and
\item
should act on $\credatm$ by direct substitution (see \rulefont{\sigma a} and \rulefont{\sigma b} in Figure~\ref{fig.sigma}), and
\item
should satisfy the equalities in Figure~\ref{fig.nomalg}.
\end{itemize*}
\end{thrm}
\begin{proof}
Immediate from the definitions and lemmas thus far, which were designed to verify these properties.
\end{proof}

\begin{rmrk}
There is redundancy in Figure~\ref{fig.nomalg}.  
For instance, a nominal algebra that satisfies \rulefont{\sigma\alpha} satisfies \rulefont{\sigma id} if and only if it satisfies \rulefont{\sigma ren}.
One half of this implication is implicit in the proof of Lemma~\ref{lemm.sigma.ren}, which derives \rulefont{\sigma ren} from (the lemmas corresponding to) \rulefont{\sigma\alpha} and \rulefont{\sigma id}; going in the other direction is no harder.

Likewise \rulefont{\sigma swp} can be derived from \rulefont{\sigma\sigma} and \rulefont{\sigma\#}.
This does no harm: in this paper we are interested in exploring the good properties of Definition~\ref{defn.sigma}, rather than studying minimal sets of axioms for their own sake (for which see \cite{gabbay:capasn-jv}).
\end{rmrk}

\begin{rmrk}
We do not demand that the sigma-action should distribute through $\credelt$ but this is because this is syntactically impossible: $\myelt{y}{a}[a\sm x]$ cannot be equal to $\myelt{y[a\sm x]}{x}$ because $\myelt{y[a\sm x]}{x}$ is not syntax according to Figure~\ref{fig.int.syntax}.

We shall see in Subsection~\ref{subsect.sugar}, however, that this all works after all, in a suitable sense, and this will become an important observation when interpreting TST in internal syntax in Section~\ref{sect.tst}.
\end{rmrk}

\begin{rmrk}
Another way to approach the proofs in this paper would be to admit $\myelt{y}{x}$ and an explicit substitution term-former, and orient Figure~\ref{fig.sigma} as rewrite rules.
We would obtain a \emph{nominal rewrite system} \cite{gabbay:nomr-jv}.
Essentially this would amount to converting Figure~\ref{fig.sigma} (and the proofs that use it) to a `small-step' presentation, from the current `big-step' form.
\end{rmrk}

\subsection{The sugar $y\tin x$ and its properties}
\label{subsect.sugar}

Figure~\ref{fig.int.syntax} only permits the syntax $\myelt{y}{a}$, not the syntax $\myelt{y}{x}$.
We can obtain the power of $\myelt{y}{x}$ via a more sophisticated operation which we construct out of components already available:
\begin{nttn}
\label{nttn.tin}
\label{nttn.btina}
\begin{itemize}
\item
Suppose $i{\in}\integer$ and $x{\in}\lset{i}$ is an internal comprehension\footnote{Terminology from Notation~\ref{nttn.internal.comprehension}. So $x$ has the form $\myst{b}{x\at b}$ for some $b{\in}\atomsi{i\minus 1}$ with $b\#x$, and $x$ does not have the form $\myatm{a}$ for any $a{\in}\atomsi{i1}$.} and $y{\in}\lset{i\minus 1}$.
Then define $y\tin x$ by
$$
y\tin x = (x\at b)[b\sm y]
$$
where we choose $b{\in}\atomsi{i\minus 1}$ fresh (so $b\#x,y$).
\item
Suppose $i{\in}\integer$ and $a{\in}\atomsi{i}$ and $y{\in}\lset{i\minus 1}$.
Then define $y\tin\myatm{a}$ and $y\tin a$ by
$$
y\tin \myatm{a} = \myelt{y}{a} = y\tin a.
$$
\end{itemize}
\end{nttn}

\begin{rmrk}
Two natural sanity properties for Notation~\ref{nttn.tin} are that 
\begin{enumerate*}
\item
it should interact well with sets comprehension on the right-hand side, and 
\item
it should interact well with the sigma-action substituting variables for terms.
\end{enumerate*}
This is Lemmas~\ref{lemm.Xax.tin} and~\ref{lemm.tin.sigma}.
\end{rmrk}

\begin{lemm}
\label{lemm.Xax.tin} 
Suppose $X{\in}\cred$ and $i{\in}\integer$ and $a{\in}\atomsi{i}$ and $x{\in}\lset{i}$ and $a\#x$.
Then (using Notation~\ref{nttn.tin})
$$
x\tin\myst{a}{X} = X[a\sm x]
.
$$
\end{lemm}
\begin{proof}
Note by Lemma~\ref{lemm.supp.abs} that $a\#\myst{a}{X}$.
By Notation~\ref{nttn.tin} (since $a\#x,\myst{a}{X}$) $x\tin\myst{a}{X}$ is equal to $((\myst{a}{X})\at a)[a\sm x]$ and by Lemma~\ref{lemm.conc.set}(4) 
this is equal to $X[a\sm x]$.
\end{proof}

\begin{lemm}
\label{lemm.tin.sigma}
Suppose $i,j{\in}\integer$ and $x{\in}\lset{i{+}1}$ and $y{\in}\lset{i}$ and $a{\in}\atomsj{j}$ and $u{\in}\lset{j}$. 
Then:
\begin{enumerate}
\item\label{tin.distrib}
$(y\tin x)[a\sm u]=y[a\sm u]\tin x[a\sm u]$.
\item\label{tin.distrib.a}
$(y\tin a)[a\sm u]=y[a\sm u]\tin u$, where $j=i\plus 1$.
\item\label{tin.distrib.b}
$(y\tin b)[a\sm u]=y[a\sm u]\tin b$, where $b{\in}\atomsi{i\plus 1}$.
\end{enumerate}
\end{lemm}
\begin{proof}
First, suppose we have proved part~\ref{tin.distrib} of this result.
Then part~\ref{tin.distrib.a} follows using \figref{fig.sigma}{\sigma a} and part~\ref{tin.distrib.b} follows using \figref{fig.sigma}{\sigma b}.

To prove part~\ref{tin.distrib} there are three cases:
\begin{itemize}
\item
\emph{Suppose $x{=}\myatm{a'}$ for some $a'{\in}\atomsi{i\plus 1}$ not equal to $a$.}\quad

We reason as follows:
$$
\begin{array}[b]{r@{\ }l@{\qquad}l}
(y\tin \myatm{a'})[a\sm u]
=&
(\myelt{y}{a'})[a\sm u]
&\text{Notation~\ref{nttn.tin}}
\\
=&
\myelt{y[a\sm u]}{a'}
&\figref{fig.sub}{\sigma\credelt b}
\\
=&
y[a\sm u]\tin \myatm{a'}
&\text{Notation~\ref{nttn.tin}}
\\
=&
y[a\sm u]\tin (\myatm{a'}[a\sm u])
&\figref{fig.sub}{\sigma b}
\end{array}
$$
\item
\emph{Suppose $x{=}\myatm{a}$ (so that $j{=}i\plus 1$).}\quad

There are two sub-cases:
\begin{itemize}
\item
\emph{Suppose $u{=}\myatm{n}$ for some $n{\in}\atomsi{i}$.}\quad
We reason as follows:
$$
\begin{array}{r@{\ }l@{\qquad}l}
(y\tin \myatm{a})[a\sm \myatm{n}]
=&
(\myelt{y}{a})[a\sm \myatm{n}]
&\text{Notation~\ref{nttn.tin}}
\\
=&
\myelt{y[a\sm \myatm{n}]}{n}
&\figref{fig.sub}{\sigma\credelt\atm}
\\
=&
y[a\sm\myatm{n}]\tin \myatm{n}
&\text{Notation~\ref{nttn.tin}}
\end{array}
$$
\item
\emph{Suppose $u{=}\myst{a'}{U'}$ where $U'=u\at a'$ for some fresh $a'{\in}\atomsi{i\minus 1}$ (so $a'\#u,y$).}\quad
We reason as follows:
$$
\begin{array}{r@{\ }l@{\qquad}l}
(y\tin \myatm{a})[a\sm u]
=&
(\myelt{y}{a})[a\sm u]
&\text{Notation~\ref{nttn.tin}}
\\
=&
U'[a'\sm y[a\sm u]]
&\figref{fig.sub}{\sigma\credelt a}
\\
=&
y[a\sm u]\tin u
&\text{Notation~\ref{nttn.tin}}
\end{array}
$$
\end{itemize}
\item
\emph{Suppose $x$ is an internal comprehension (not an internal atom).}\quad

Choose $b{\in}\atomsi{i}$ fresh (so $b\#x,y,u$).
We reason as follows:
\begin{align*}
\begin{array}[b]{r@{\ }l@{\qquad}l}
(y\tin x)[a\sm u]
=&
(x\at b)[b\sm y][a\sm u]
&\text{Notation~\ref{nttn.tin}}
\\
=&
(x\at b)[a\sm u][b\sm y[a\sm u]]
&\text{Proposition~\ref{prop.sigma.sigma}}\ b\#u
\\
=&
(x[a\sm u]\at b)[b\sm y[a\sm u]]
&\text{Lemma~\ref{lemm.sigma.at.b}}\ b\#x,u
\\
=&
y[a\sm u]\tin x[a\sm u] 
&\text{Notation~\ref{nttn.tin}}\tag*{\qEd}
\end{array}
\end{align*}
\end{itemize}
\def\popQED{}
\end{proof}

\begin{corr}
\label{corr.tin.tin}
Suppose $k{\in}\integer$ and $c{\in}\atomsi{k}$ and $x{\in}\lset{k}$ and $y{\in}\lset{k\minus 1}$ and $c\#y$.
Then 
$$
\text{if}\ \ z{=}\myst{c}{(y\tin c)}{\in}\lset{k{+}1} \ \ \text{then}\ \ 
x\tin z = y\tin x .
$$
\end{corr}
\begin{proof}
We reason as follows:
\begin{align*}
\begin{array}[b]{r@{\ }l@{\qquad}l}
x\tin z =& x\tin\myst{c}{y\tin c}
&\text{Assumption}
\\
=& (y\tin c)[c\sm x]
&\text{Lemma~\ref{lemm.Xax.tin}} 
\\
=& y[c\sm x]\tin x
&\text{Lemma~\ref{lemm.tin.sigma}(\ref{tin.distrib})}
\\
=& y\tin x
&\text{Lemma~\ref{lemm.sigma.hash.always}}\,c\#y\tag*{\qEd}
\end{array}
\end{align*}
\def\popQED{}
\end{proof}

\begin{rmrk}
It is quite interesting to reflect on the inductive measures used in the proofs above.
Collecting them a list, they are:
\begin{itemize*}
\item
$\age(Z)$ and $\age(z)$ in 
Lemmas~\ref{lemm.sigma.alpha},\ \ref{lemm.sigma.hash.always},\ and~\ref{lemm.sigma.id}. 
\item
$(\level(a),\age(Z))$ and $(\level(a),\age(z))$ in Proposition~\ref{prop.X.sigma}. 
\item
$(\level(a){+}\level(b),\age(Z))$ and $(\level(a){+}\level(b),\age(z))$ in Proposition~\ref{prop.sigma.sigma}.
\end{itemize*}
So we see that the inductive proofs fall into two categories: 
\begin{enumerate*}
\item
those inductive proofs that are by a direct induction on structure and are `fairly simple', and 
\item
those inductive proofs that depend on the hierarchy of levels and are `slightly harder'.
\end{enumerate*}
Looking deeper at the slightly harder results, the inductive quantities seem to follow a slogan of
\begin{quote}
\emph{take the sum of the levels of relevant atoms, and the age of the relevant terms, lexicographically ordered.}
\end{quote}
These inductive quantities are simple, though a certain amount of thinking was required to develop them in the first place.
For future work, if e.g. a package for handling stratified syntax is implemented in a theorem-prover, then the slogan above might form the basis of a generic automated proof-method. 

Though the proofs above are probably susceptible to automation by a sufficiently advanced tactic, they are not all the same.
\end{rmrk}

\section{The language of Typed Sets}
\label{sect.tst}

We now have everything we need to develop the syntax of Typed Set Theory.
 
\subsection{Syntax of Stratified Sets}
\label{subsect.language}

\begin{defn}
\label{defn.sets.syntax}
Let \deffont{formulae} and \deffont{terms} be inductively defined as in Figure~\ref{fig.NF.syntax}.
In that figure, $a$ ranges over atoms (Definition~\ref{defn.atoms}) and $\tall$ is taken to bind $a$ and we quotient by $\alpha$-equivalence.
We write $\phi[a\ssm s]$ and $r[a\ssm s]$ for the usual capture-avoiding substitution on syntax.
\end{defn}

\begin{rmrk}
\label{rmrk.usual.syntax}
Quotienting by $\alpha$-equivalence means that a formula $\phi$ is actually an $\alpha$-equivalence class of syntax-trees and similarly for a term $r$.
This is a typical treatment but we could just as easily set things up differently, e.g. using nominal abstract syntax or de Bruijn indexes.
Definition~\ref{defn.sets.syntax} as written is designed to be close to what one might find in a typical paper on TST+ or NF if the syntax were specified.\footnote{\dots which it typically is not.  For instance~\cite{forster:settus} does not formally define its syntax e.g. via an inductive definition in the style of Definition~\ref{defn.sets.syntax}, and this is not unusual.} 
\end{rmrk}

\begin{figure}
$$
\begin{array}{r@{\ }l}
\phi,\psi ::=& \tbot \mid \tneg\phi \mid \phi\tand\phi \mid \tall a.\phi \mid 
s\tin s
\\
s,t,r::=& a \mid \tst{a}{\phi}
\end{array}
$$
\caption{The syntax of Stratified Sets}
\label{fig.NF.syntax}
\end{figure}

Definition~\ref{defn.levels} is standard: 
\begin{defn}
\label{defn.levels}
Suppose $t$ is a term (Definition~\ref{defn.sets.syntax}).
Then extend $\level(a)$ from Definition~\ref{defn.atoms} from atoms to all terms by:
$$
\begin{array}{r@{\ }l}
\level(\tst{a}{\phi})=&\level(a){+}1
\end{array}
$$
Call a formula $\phi$ or term $t$ \deffont{stratified} when:
\begin{quote}
if $s'\tin s$ is a subterm of $\phi$ or $t$ then $\level(s)=\level(s'){+}1$.
\end{quote}
\end{defn}

\begin{xmpl}
Suppose $a{\in}\atomsi{2}$, $b{\in}\atomsi{3}$, and $c{\in}\atomsi{4}$.
Then $a\tin b$ and $b\tin c$ are stratified, and $a\tin c$, $b\tin a$, and $a\tin a$ are not stratified.
\end{xmpl}

\begin{defn}
\label{defn.tst}
The language of \deffont{Stratified Sets} consists of stratified formulae and terms.
\end{defn}

\begin{rmrk}
We only care about \emph{stratified} formulae and terms henceforth --- that is, we restrict attention to those formulae and terms that are stratified.

So for all terms and formulae considered from now on, the reader should assume they are stratified, where:
\begin{itemize*}
\item
$\tin$ polymorphically takes two terms of type $i$ and $i{+}1$ to a formula for each $i\in\integer$,
and
\item
sets comprehension $\tst{a}{\phi}$ takes an atom of type $i\in\integer$ and a formula $\phi$ to a term of type $i{+}1$. 
\end{itemize*}
We further assume that levels are arranged to respect stratification where this is required, so for example when we write $[a\ssm s]$ it is understood that we assume $a{\in}\atomsi{\level(s)}$.
\end{rmrk}

\begin{rmrk}
We could add equality $s\teq t$ to our syntax in Figure~\ref{fig.NF.syntax}, at some modest cost in extra cases in inductive arguments.
The pertinent stratification condition would be that if $s\teq t$ is a subterm then $\level(s)=\level(t)$.
Our results extend without issues to the syntax with equality.
\end{rmrk}

\subsection{Interpretation for formulae and terms}

\begin{defn}
\label{defn.interp}
Define an \deffont{interpretation} of stratified formulae $\phi$ and terms $s$ as in Figure~\ref{fig.interp}, mapping $\phi$ to $\finte{\phi}\in\cred$ and $s$ of level $i{\in}\integer$ to $\finte{s}\in\lset{i}$. 
\end{defn}

\begin{rmrk}
For the reader's convenience we give pointers for the notation used in the right-hand sides of the equalities in Figure~\ref{fig.interp}: 
\begin{itemize*}
\item
$\credfalse$ is from Example~\ref{xmpl.credemp}.
\item
$\credneg{}$ is from Definition~\ref{defn.cred}.
\item
$\finte{t}\tin\finte{s}$ is from Notation~\ref{nttn.tin}.
\item
$\myst{a}{\finte{\phi}}$ is from Definitions~\ref{defn.abs} and~\ref{defn.cred}. 
\item
$\atm$ is from Definition~\ref{defn.lsets}.
\end{itemize*}
Note that the translation in Figure~\ref{fig.interp} from the syntax of formulae $\phi$ and terms $s$ from Figure~\ref{fig.NF.syntax} to the syntax of internal predicates and internal sets from Definition~\ref{defn.cred} is not entirely direct: $t\tin s$ is primitive in formulae but only primitive in internal predicates if $s$ is an atom.
\end{rmrk}

\begin{figure}[t]
$$
\begin{array}[t]{r@{\ }l}
\finte{\tbot}=&\credfalse
\\
\finte{\tneg\phi}=&\myneg{\finte{\phi}}
\\
\finte{\phi\tand\psi}=&\mybinaryand{\finte{\phi}}{\finte{\psi}}
\\
\finte{\tall a.\phi}=&\myall{a}{\finte{\phi}}
\end{array}
\qquad
\begin{array}[t]{r@{\ }l}
\finte{t\tin s}=&\finte{t}\tin\finte{s}
\\
\finte{\tst{a}{\phi}}=&\myst{a}{\finte{\phi}} 
\\
\finte{a}=&
\myatm{a} 
\end{array}
$$
\caption{Interpretation of formulae and terms}
\label{fig.interp}
\end{figure}

\begin{defn}
\label{defn.size}
Define the \deffont{size} of a stratified formula $\phi$ and stratified term $t$ inductively as follows:
$$
\begin{array}{r@{\ }l@{\quad}r@{\ }l}
\size(a)=&1
&
\size(\tst{a}{\phi})=&\size(\phi){+}1
\\
\size(\tbot)=&1
&
\size(\phi\tand\psi)=&\size(\phi)+\size(\psi)+1
\\
\size(\tneg\phi)=&\size(\phi){+}1
&
\size(\tall a.\phi)=&
\size(\phi){+}1
\\
\size(t\tin s)=&\size(t){+}\size(s)+1
\end{array}
$$
\end{defn}

\begin{lemm}
\label{lemm.level.denot.s}
Suppose $\phi$ is a stratified formula and $s$ is a stratified term with $\level(s)= i\in\integer$.
Then 
$$
\finte{\phi}\in\cred
\quad\text{and}\quad
\finte{s}\in\lset{i}.
$$
\end{lemm}
\begin{proof}
By induction on $\size(\phi)$ and $\size(s)$:\footnote{A structural induction on nominal abstract syntax~\cite{gabbay:newaas-jv} would also work, and in a nominal mechanised proof might be preferable.  Similarly for Proposition~\ref{prop.interp.sub}.} 
\begin{itemize*}
\item
\emph{The case of $a$.}\quad
By Figure~\ref{fig.interp} $\finte{a}{=}\myatm{a}$.
By Definition~\ref{defn.lsets} $\myatm{a}{\in}\lset{i}$.
\item
\emph{The case of $\tst{b}{\phi}$ for $j{\geq}1$ and $b{\in}\atomsj{j}$.}\quad 
By Figure~\ref{fig.interp} $\finte{\tst{b}{\phi}}{=}\myst{b}{\finte{\phi}}$. 
By Definition~\ref{defn.levels} $\level{\tst{b}{\phi}}=j{+}1$.
By inductive hypothesis $\finte{\phi}{\in}\cred$ and by Definition~\ref{defn.lsets} $\myst{b}{\finte{\phi}}\in\lset{j{+}1}$.
\item
\emph{The case of $\tbot$.}\quad
By Figure~\ref{fig.interp} $\finte{\tbot}{=}\credfalse\in\cred$. 
\item
\emph{The case of $\tneg\phi$.}\quad
From Figure~\ref{fig.interp} and Definition~\ref{defn.cred} using the inductive hypothesis. 
\item
\emph{The case of $\phi\tand\psi$.}\quad
From Figure~\ref{fig.interp} and Definition~\ref{defn.cred} using the inductive hypothesis. 
\item
\emph{The case of $\tall a.\phi$.}\quad
From Figure~\ref{fig.interp} and Definition~\ref{defn.cred} using the inductive hypothesis. 
\item
\emph{The case of $t\tin s$.}\quad
We refer to Notation~\ref{nttn.tin} and use Lemma~\ref{lemm.conc.set} and Proposition~\ref{prop.X.sigma}.
\qedhere\end{itemize*}
\end{proof}

\subsection{Properties of the interpretation}

\begin{prop}
\label{prop.interp.sub}
Suppose $\phi$ is a stratified formula and $t$, and $r$ are stratified terms and $b{\in}\atomsi{\level(t)}$.
Then: 
$$
\begin{array}{r@{\ }l}
\finte{\phi}[b\sm\finte{t}]
=&
\finte{\phi[b\ssm t]} 
\\
\finte{r}[b\sm\finte{t}]
=&
\finte{r[b\ssm t]}
\end{array}
$$
\end{prop}
\noindent Note by Lemma~\ref{lemm.level.denot.s} that $\finte{t}{\in}\lset{\level(t)}$ so that the $\sigma$-action $[b\sm\finte{t}]$ above is well-defined (Definition~\ref{defn.sub}).
\begin{proof}
By induction on $\size(\phi)$ and $\size(r)$. 
We consider each case in turn; the interesting case is for $\tin$, where we use Lemma~\ref{lemm.tin.sigma}:
\begin{itemize*}
\item
\emph{The case of $\tbot$.}\quad
We reason as follows:
$$
\begin{array}{r@{\ }l@{\qquad}l}
\finte{\tbot}[b\sm\finte{t}]
=&\credfalse[b\sm\finte{t}]
&\text{Figure~\ref{fig.interp}}
\\
=&\credfalse
&\text{Corollary~\ref{corr.credemp.sigma}}
\\
\finte{\tbot[b\ssm t]}
=&\finte{\tbot}
&\text{Fact of syntax}
\\
=&\credfalse
&\text{Figure~\ref{fig.interp}}
\end{array}
$$
\item
\emph{The case of $\tneg\phi$.}\quad
We reason as follows:
$$
\begin{array}{r@{\ }l@{\quad}l}
\finte{\tneg\phi}[b\sm\finte{t}]
=&(\myneg{\finte{\phi}})[b\sm\finte{t}]
&\text{Figure~\ref{fig.interp}}
\\
=&\myneg{\finte{\phi}[b\sm\finte{t}]}
&\text{Figure~\ref{fig.sub}} 
\\
=&\myneg{\finte{\phi[b\ssm t]}}
&\text{IH} \ \size(\phi){<}\size(\tneg\phi)
\\
=&\finte{\tneg(\phi[b\ssm t])}
&\text{Figure~\ref{fig.sub}} 
\\
=&\finte{(\tneg\phi)[b\ssm t]}
&\text{Fact of syntax}
\end{array}
$$
\item
\emph{The case of $\phi\tand\psi$.}\quad
We reason as follows:
$$
\begin{array}{r@{\ }l@{\quad}l}
\finte{\phi\tand\psi}[b\sm\finte{t}]
=&(\mybinaryand{\finte{\phi}}{\finte{\psi}})[b\sm\finte{t}]
&\text{Figure~\ref{fig.interp}}
\\
=&\mybinaryand{(\finte{\phi}[b\sm\finte{t}])}{(\finte{\psi}[b\sm\finte{t}])}
&\text{Figure~\ref{fig.sub}} %
\\
=&\mybinaryand{\finte{\phi[b\ssm t]}}{\finte{\psi[b\ssm t]}}
&\text{IH} \ \size(\phi),\size(\psi){<}\size(\phi\tand\psi)
\\
=&\finte{\phi[b\ssm t]\tand(\psi[b\ssm t])}
&\text{Figure~\ref{fig.sub}} 
\\
=&\finte{(\phi\tand\psi)[b\ssm t]}
&\text{Fact of syntax}
\end{array}
$$
\item
\emph{The case of $\tall a.\phi$.}\quad
We reason as follows, where we $\alpha$-rename if necessary to assume $a\#t$ (from which it follows by Theorem~\ref{thrm.no.increase.of.supp} that $a\#\finte{t}$):
$$
\begin{array}{r@{\ }l@{\qquad}l}
\finte{\tall a.\phi}[b\sm\finte{t}]
=&\myall{a}{\finte{\phi}}[b\sm\finte{t}]
&\text{Figure~\ref{fig.interp}}
\\
=&\myall{a}{(\finte{\phi}[b\sm\finte{t}])}
&\text{Figure~\ref{fig.sub}}\ a\#\finte{t}
\\
=&\myall{a}{\finte{\phi[b\ssm t]}}
&\text{IH} \ \size(\phi){<}\size(\tall a.\phi)
\\
=&\finte{\tall a.(\phi[b\ssm t])}
&\text{Figure~\ref{fig.interp}}
\\
=&\finte{(\tall a.\phi)[b\ssm t]}
&\text{Fact of syntax},\ a\#t
\end{array}
$$
\item
\emph{The case of $b$.}\quad
By Figure~\ref{fig.interp} $\finte{b}{=}\myatm{b}$.
By assumption $\finte{t}{\in}\lset{\level(b)}$ so by \figref{fig.sub}{\sigma a}
$$
\myatm{b}[b\sm \finte{t}]=\finte{t}.
$$
\item
\emph{The case of $a$ (any atom other than $b$).}\quad
By Figure~\ref{fig.interp} $\finte{a}{=}\myatm{a}$.
We use rule \rulefont{\sigma b} of Figure~\ref{fig.sub}.
\item
\emph{The case of $\tst{a}{\phi}$.}\quad
$\alpha$-converting if necessary assume $a$ is fresh (so $a\#t$, and by Theorem~\ref{thrm.no.increase.of.supp} also $a\#\finte{t}$).
We reason as follows:
$$
\begin{array}{r@{\ }l@{\qquad}l}
\finte{\tst{a}{\phi}}[b\sm\finte{t}]
=&(\myst{a}{\finte{\phi}})[b\sm \finte{t}]
&\text{Figure~\ref{fig.interp}}
\\
=&\myst{a}{\finte{\phi}[b\sm\finte{t}]}
&\figref{fig.sub}{\sigma\credst},\ a\#\finte{t}
\\
=&\myst{a}{\finte{\phi[b\ssm t]}}
&\text{IH} \ \size(\phi){<}\size(\tst{a}{\phi})
\\
=&\finte{\tst{a}{\phi[b\ssm t]}}
&\text{Figure~\ref{fig.interp}},\ a\#t
\\
=&\finte{\tst{a}{\phi}[b\ssm t]} 
&\text{Fact of syntax}
\end{array}
$$
\item
\emph{The case of $t'\tin s'$.}\quad
We reason as follows:
\begin{align*}
\begin{array}[b]{r@{\ }l@{\quad}l}
\finte{t'\tin s'}[b\sm \finte{t}]
=&
(\finte{t'}\tin\finte{s'})[b\sm\finte{t}]
&\text{Figure~\ref{fig.interp}}
\\
=&
(\finte{t'}[b\sm\finte{t}])\tin(\finte{s'}[b\sm\finte{t}])
&\text{Lemma~\ref{lemm.tin.sigma}}
\\
=&
(\finte{t'[b\ssm t]})\tin(\finte{s'[b\ssm t]})
&\text{IH} \ \size(t'),\size(s'){<}\size(t'\tin s')
\\
=&
\finte{(t'[b\ssm t])\tin (s'[b\ssm t])}
&\text{Figure~\ref{fig.interp}}\tag*{\qEd}
\end{array}
\end{align*}
\end{itemize*}
\def\popQED{}
\end{proof}

\begin{lemm}
\label{lemm.comprehension}
\label{lemm.finte.to.sm}
Suppose $\phi$ is a stratified formula and $s$ is a stratified term.
Suppose $a\in\atomsi{i{+}1}$ and $\level(s)=i$.
Then:
\begin{enumerate}
\item
$\finte{s\tin\tst{a}{\phi}}=\finte{\phi[a\ssm s]}$.
\item
$\finte{s\tin\tst{a}{\phi}} = \finte{\phi}[a\sm\finte{s}]$.
\end{enumerate}
\end{lemm}
\begin{proof}
We reason as follows:
\begin{align*}
\begin{array}[b]{r@{\ }l@{\qquad}l}
\finte{s\tin\tst{a}{\phi}}
=&
\finte{s}\tin\finte{\tst{a}{\phi}}
&\text{Figure~\ref{fig.interp}}
\\
=&
\finte{s}\tin\myst{a}{\finte{\phi}}
&\text{Figure~\ref{fig.interp}}
\\
=&
(\myst{a}{\finte{\phi}}\at a)[a\sm\finte{s}]
&\text{Notation~\ref{nttn.tin}}
\\
=&
\finte{\phi}[a\sm\finte{s}]
&\text{Lemma~\ref{lemm.conc.set}(4)}
\\
=&
\finte{\phi[a\ssm s]}
&\text{Proposition~\ref{prop.interp.sub}}\tag*{\qEd}
\end{array}
\end{align*}
\def\popQED{}
\end{proof}

\subsection{Confluence}

\begin{defn}
\label{defn.to}
\begin{enumerate}
\item
Let $\to$ be a rewrite relation on the language of Stratified Sets (Definition~\ref{defn.tst})
defined by the rules in Figure~\ref{fig.rewrite}.
\item
Write $\to^*$ for the transitive reflexive closure of $\to$ (so the least transitive reflexive relation containing $\to$).
\end{enumerate}
\end{defn}

\begin{figure}[t]
$$
\begin{array}{c@{\qquad}c@{\qquad}c@{\qquad}c}
\begin{prooftree}
\phantom{h}
\justifies
t\tin \tst{a}{\phi}\to \phi[a\ssm t] 
\end{prooftree}
&
\begin{prooftree}
\phi\to\phi'
\justifies
\tneg\phi\to\tneg\phi'
\end{prooftree}
& 
\begin{prooftree}
\phi\to\phi'
\justifies
\phi\tand\phi''\to\phi'\tand\phi''
\end{prooftree}
&
\begin{prooftree}
\phi\to\phi'
\justifies
\phi''\tand\phi\to\phi''\tand\phi'
\end{prooftree}
\\[4ex]
\begin{prooftree}
\phi\to\phi'
\justifies
\tst{a}{\phi}\to\tst{a}{\phi'}
\end{prooftree}
&
\begin{prooftree}
\phi\to\phi'
\justifies
\tall a.\phi\to\tall a.\phi'
\end{prooftree}
&
\begin{prooftree}
s\to s'
\justifies
t\tin s\to t\tin s'
\end{prooftree}
&
\begin{prooftree}
t\to t'
\justifies
t\tin s\to t'\tin s
\end{prooftree}
\end{array}
$$
\caption{Rewrite system on formulae and terms}
\label{fig.rewrite}
\end{figure}

\begin{nttn}
\label{nttn.inclusive}
The natural injection of internal predicates and internal terms into stratified formulae and terms is clear;
to save notation we elide it, thus effectively treating the syntax of internal predicates and terms from Definition~\ref{defn.cred} as a direct subset of the syntax of stratified formulae and terms from Figure~\ref{fig.NF.syntax}.
The reader who dislikes this abuse of notation can fill in an explicit injection function $\iota$ as required, to map the former injectively into the latter.  
Either way, the meaning will be the same.
\end{nttn}

\begin{nttn}
\label{nttn.reduct}
Call a formula of the form $t\tin\tst{a}{\phi}$ a \deffont{reduct}.
\end{nttn}

\begin{lemm}
\label{lemm.finte.normal.form} 
$\finte{\phi}$ considered as a formula, is a $\to$-normal form, and similarly for $\finte{s}$. 
\end{lemm}
\begin{proof}
Reducts are impossible because the internal syntax from Figure~\ref{fig.int.syntax} only allows us to form $y\tin x$ (written $\myelt{y}{x}$ in that figure) when $x$ is an atom and not a comprehension.
\end{proof}

We can now state a kind of converse to Lemma~\ref{lemm.finte.normal.form}: 
\begin{thrm}
\label{thrm.confluent}
\begin{enumerate*}
\item
$\phi\to^*\finte{\phi}$ and $s\to^*\finte{s}$.

Here we use Notation~\ref{nttn.inclusive} to treat $\finte{\phi}$ and $\finte{s}$ directly as stratified syntax.
\item
If $\phi\to\phi'$ then $\finte{\phi}=\finte{\phi'}$.
If $s\to s'$ then $\finte{s}=\finte{s'}$.
\item
$\phi\to^*\phi'$ then $\phi'\to^*\finte{\phi}$, and if $s\to^*s'$ then $s'\to^*\finte{s'}$. 
As a corollary, the rewrite relation $\to$ from Figure~\ref{fig.rewrite} is confluent. 
\end{enumerate*}
\end{thrm} 
\begin{proof}
\begin{enumerate}
\item
By induction on syntax.
The interesting case is for $t\tin\tst{a}{\phi}$.
Suppose $\phi\to^*\finte{\phi}$.
Then
$$
\begin{array}{r@{\ }l@{\qquad}l}
t\tin\tst{a}{\phi} 
\to^*& \finte{t}\tin\tst{a}{\phi}
&\text{IH, Figure~\ref{fig.rewrite}}
\\
\to& \finte{\phi}[a\sm \finte{t}]
&\text{Figure~\ref{fig.rewrite}}
\\
=&\finte{t\tin\tst{a}{\phi}}
&\text{Lemma~\ref{lemm.finte.to.sm}}
\end{array}
$$
\item
By induction on the derivation of the rewrite (that is, on the term-context in which the rewrite takes place).
We consider three cases:
\begin{itemize}
\item
Suppose $t\tin\tst{a}{\phi}\to\phi[a\ssm t]$.
By Lemma~\ref{lemm.finte.to.sm} $\finte{t\tin\tst{a}{\phi}}=\finte{\phi[a\ssm t]}$.
\item
Suppose $t\tin\tst{a}{\phi}\to t\tin\tst{a}{\phi'}$ because $\tst{a}{\phi}\to\tst{a}{\phi'}$ because $\phi\to\phi'$.
By induction hypothesis $\finte{\phi}=\finte{\phi'}$.
Then using Lemma~\ref{lemm.finte.to.sm} we have 
$$
\finte{t\tin \tst{a}{\phi}} 
\stackrel{L\ref{lemm.finte.to.sm}}= 
\finte{\phi}[a\sm \finte{t}] 
= 
\finte{\phi'}[a\sm\finte{t}] 
\stackrel{L\ref{lemm.finte.to.sm}}= 
\finte{t\tin\tst{a}{\phi}}
.
$$
\item
Suppose $t\tin\tst{a}{\phi}\to t'\tin\tst{a}{\phi}$ because $t\to t'$.
By induction hypothesis $\finte{t}=\finte{t'}$.
We have
$$
\finte{t\tin \tst{a}{\phi}} 
\stackrel{L\ref{lemm.finte.to.sm}}= 
\finte{\phi}[a\sm \finte{t}] 
=
\finte{\phi}[a\sm \finte{t'}] 
\stackrel{L\ref{lemm.finte.to.sm}}= 
\finte{t'\tin\tst{a}{\phi}}
.
$$
\end{itemize}
\item
We combine parts~1 and~2 of this result.
\qedhere\end{enumerate}
\end{proof}

\subsection{Strong normalisation}

\begin{nttn}
\begin{itemize*}
\item
Call a comprehension $\tst{a}{\phi}$ \deffont{ternary} when $a$ occurs in $\phi$ \emph{at least} three times.\footnote{This is arguably an abuse of notation; `ternary' might suggest \emph{exactly} three times. But we need a name.}

If $a$ occurs in $\phi$ zero, one, or two times then we call $\tst{a}{\phi}$ \deffont{non-ternary}.
\item
Call a formula $\phi$ \deffont{ternary} if every comprehension in it is ternary. 
\item
Call a term $t$ \deffont{ternary} if every comprehension in it is ternary. 
\end{itemize*}
\end{nttn}

\begin{defn}
\label{defn.pad}
Suppose $\phi$ is a formula and $t$ is a term.
Write $\f{na}(\phi)$ and $\f{na}(t)$ for the predicate or term obtained by padding every non-ternary comprehension $\tst{a'}{\phi'}$ in $\phi$ or $t$ to 
$$
\tst{a'}{\phi'\tand\texi c.(a'\tin c\tand a'\tin c\tand a'\tin c)}
$$ 
where $c\in\atomsi{\level(a')\plus 1}$.\footnote{The precise choice of $c$ does not matter since we bind it with $\texi$.  The level matters, so the result is stratified.} 
\end{defn}

\begin{lemm}
\label{lemm.affine}
Suppose $\phi$ is a formula and $t$ is a term.
Then:
\begin{enumerate*}
\item
$\f{na}(\phi)$ and $\f{na}(t)$ are ternary.
\item
If $\phi$ and $s$ are ternary then $\phi[a\ssm s]$ is ternary.
Similarly for $t[a\ssm s]$.
\item
If $\phi$ is ternary and $\phi\to\phi'$ then $\phi'$ is ternary.
Similarly for $t\to t'$.
\end{enumerate*}
\end{lemm}
\begin{proof}
\begin{enumerate}
\item
By construction.
\item
By an easy calculation.
\item
By an easy calculation from Figure~\ref{fig.rewrite} and part~2 of this result.
\qedhere\end{enumerate}
\end{proof}

\begin{defn}
\label{defn.complexity}
Define the \deffont{complexity} of a stratified formula $\phi$ and stratified term $t$ as follows:
$$
\begin{array}{r@{\ }l@{\quad}l}
\complexity(a)=&1
\\
\complexity(\tst{a}{\phi})=&1+\complexity(\phi)
\\
\complexity(\tbot)=&3
\\
\complexity(\phi\tand\psi)=&\complexity(\phi)+1+\complexity(\psi)
\\
\complexity(\tneg\phi)=&1+\complexity(\phi)
\\
\complexity(\tall a.\phi)=&1+\complexity(\phi)
\\
\complexity(b\tin\tst{a}{\phi})=&\complexity(\phi)
\\
\complexity(t\tin s)=&\complexity(t)+1+\complexity(s)
&\begin{array}[t]{l}\text{$t$ not an atom, or}\\ \text{$s$ not a comprehension}\end{array}
\end{array}
$$
Define the \deffont{number of atomic reducts} of $\phi$ and $t$ as follows:
$$
\begin{array}{r@{\ }l@{\quad}l}
\f{atomic}(a)=&0
\\
\f{atomic}(\tst{a}{\phi})=&\f{atomic}(\phi)
\\
\f{atomic}(\tbot)=&0
\\
\f{atomic}(\phi\tand\psi)=&\f{atomic}(\phi)+\f{atomic}(\psi)
\\
\f{atomic}(\tneg\phi)=&\f{atomic}(\phi)
\\
\f{atomic}(\tall a.\phi)=&\f{atomic}(\phi)
\\
\f{atomic}(t\tin b)=&\f{atomic}(t)
\\
\f{atomic}(t\tin\tst{a}{\phi})=&\f{atomic}(\phi)+1
&t\text{ an atom}
\\
\f{atomic}(t\tin s)=&\f{atomic}(t)+\f{atomic}(s)
&t\text{ not an atom}
\end{array}
$$
\end{defn}

\begin{rmrk}
\begin{enumerate}
\item
If we view $\tst{a}{\phi}$ as $\lambda a.\phi$ and $t\tin s$ as $s\,t$ then an \emph{atomic reduct} is just a term of the form $(\lambda a.\phi)b$.
Then $\f{atomic}$ counts the number of atomic reducts in formulae $\phi$ and terms $t$, and $\complexity$ is a measure of size in which atomic reducts are skipped.
Indeed, though we never use this directly, from Definition~\ref{defn.size} we see that 
$$
\size(\phi) = \complexity(\phi) + 2*\f{atomic}(\phi) ,
$$
and similarly for $t$.
\item
Intuitively, $\complexity(\tbot)=3$ means ``$\tbot$ has the same complexity as $b\tin a$''.
Technically,  $\complexity(\tbot)=3$ makes Lemma~\ref{lemm.geq.3}(2) hold, and allows the arithmetic in Lemma~\ref{lemm.complexity.geq} to go through. 
\end{enumerate}
\end{rmrk}

\begin{lemm}
\label{lemm.geq.3}
Suppose $s$ is a term and $\phi$ is a predicate.
Then:
\begin{enumerate*}
\item
$\complexity(s)\geq 1$.
\item
$\complexity(\phi)\geq 3$.
\item
If $s$ is not an atom (so it is a comprehension) then $\complexity(s)\geq 4$.
\end{enumerate*}
\end{lemm}
\begin{proof}
By calculations and induction using Definition~\ref{defn.complexity}; 
the base cases are $\tbot$ and $t\tin s$ (in particular, $b\tin a$).
\end{proof}

\begin{lemm}
\label{lemm.complexity.sub}
Suppose $\phi$ is a predicate and $t$ is a term.
Suppose $a$ is an atom and $s$ is a term and $\level(a)=\level(s)$ (so the substitution $[a\ssm s]$ is well-defined).
Then:
\begin{enumerate*}
\item
If $s$ is an atom then 
$$
\complexity(\phi[a\ssm s])=\complexity(\phi)
\quad\text{and}\quad
\complexity(t[a\ssm s])=\complexity(t)
.
$$
\item
If $s$ is not an atom (so is a comprehension)
then 
$$
\begin{array}{r@{\ }l}
\complexity(\phi[a\ssm s])\geq&\complexity(\phi)+n*(\complexity(s)\minus 1) 
\quad\text{and}\quad
\\
\complexity(t[a\ssm s])\geq&\complexity(t)+n*(\complexity(s)\minus 1)
\end{array}
$$ 
where $n$ is the number of instances of $a$ in $\phi$ or $t$.
\end{enumerate*}
\end{lemm}
\begin{proof}
\begin{enumerate}
\item
By a routine induction on $\phi$ and $t$ using Definition~\ref{defn.complexity}.
\item
Essentially this is clear because we replace $n$ instances of $a$ each with complexity $1$, with $n$ instances of $s$ each with complexity $\complexity(s)$.
The proof is by induction on syntax using Definition~\ref{defn.complexity}.
Interesting cases of the induction are $\phi$ or $t$ that do not mention $a$ (so $n=0$), $t'\tin a$, $a\tin b$, and $a\tin\tst{a'}{\phi'}$.
We consider each in turn:
\begin{itemize*}
\item
If $\phi$ or $t$ do not mention $a$ then $\phi[a\ssm s]=\phi$ and $t[a\ssm s]=t$ and the result follows.
\item
It is a fact that $\complexity(t'\tin a)=\complexity(t')+2$, and we calculate as follows:
$$
\begin{array}{r@{\ }l@{\quad}l}
\complexity(t'[a\ssm s]\tin s)
=&\complexity(t'[a\ssm s])+1+\complexity(s)
\\ 
\geq&\complexity(t')+(n\minus 1)*(\complexity(s)\minus 1)+1+\complexity(s)
\\ 
=&\complexity(t')+(n\minus 1)*(\complexity(s)\minus 1)+2+(\complexity(s)\minus 1)
\\ 
=&\complexity(t')+2+n*(\complexity(s)\minus 1)
\\ 
=&\complexity(t'\tin a)+n*(\complexity(s)\minus 1)
\end{array}
$$
\item
It is a fact that $\complexity(a\tin b)=3$, and we calculate as follows:
$$
\begin{array}{r@{\ }l@{\quad}l}
\complexity(s\tin b)
=&\complexity(s)+2
\\ 
=&3+1*(\complexity(s)\minus 1)
\\
=&\complexity(a\tin b)+1*(\complexity(s)\minus 1)
\end{array}
$$
\item
It is a fact that $\complexity(a\tin\tst{a'}{\phi'})=\complexity(\phi')$, and we calculate as follows:
\begin{align*}
\begin{array}[b]{r@{\ }l@{\quad}l}
\complexity(s\tin\tst{a'}{\phi'[a\ssm s]})
=&\complexity(s)+2+\complexity(\phi'[a\ssm s])
\\ 
\geq&\complexity(s)+2+\complexity(\phi')+(n\minus 1)*(\complexity(s)\minus 1)
\\ 
\gneq&(\complexity(s)\minus 1)+\complexity(\phi')+(n\minus 1)*(\complexity(s)\minus 1)
\\ 
=&\complexity(\phi')+n*(\complexity(s)\minus 1)
\end{array}
    \tag*{\qEd}
\end{align*}
\end{itemize*} 
\end{enumerate}
\def\popQED{}
\end{proof}

\begin{lemm}
\label{lemm.complexity.geq}
Suppose $\tst{a}{\phi}$ is ternary (so $\phi$ mentions $a$ free at least three times) and $s$ is a term and $\level(s)=\level(a)$.
Then:
\begin{enumerate*}
\item
If $s$ is an atom then
$$
\complexity(s\tin\tst{a}{\phi})=\complexity(\phi[a\ssm s]) = \complexity(\phi) .
$$ 
\item
If $s$ is not an atom then 
$$
\complexity(s\tin\tst{a}{\phi})\lneq\complexity(\phi[a\ssm s]) .
$$ 
\end{enumerate*}
\end{lemm}
\begin{proof}
\begin{enumerate}
\item
From Lemma~\ref{lemm.complexity.sub}(1).
\item
Using Definition~\ref{defn.complexity} and Lemma~\ref{lemm.complexity.sub}(2) we have the following facts:
\begin{itemize*}
\item
$\complexity(s\tin\tst{a}{\phi})=\complexity(\phi)+3+(\complexity(s)\minus 1)$.
\item
$\complexity(\phi[a\ssm s])\geq\complexity(\phi)+n*(\complexity(s)\minus 1)$.
\end{itemize*}
We can drop the $\complexity(\phi)$ on both sides and do some arithmetic:
$$
\begin{array}{r@{\ }l@{\quad}l}
3+(\complexity(s)\minus 1) 
\lneq n*(\complexity(s)\minus 1) 
\liff&
3 \lneq (n\minus 1)*(\complexity(s)\minus 1) 
\\
\liff&
3 \lneq (n\minus 1)*3
&\text{Lemma~\ref{lemm.geq.3}(3)}
\\
\liff&
2\lneq n .
\end{array}
$$
We assumed $\phi$ is ternary, which means precisely that $n\geq 3$, so this is true.
\qedhere\end{enumerate}
\end{proof}

\begin{corr}
\label{corr.affine.2}
If $\phi$ is ternary and $\phi\to\phi'$ then precisely one of the following must hold:
\begin{itemize*}
\item
$\complexity(\phi)\lneq\complexity(\phi')$ (in words: \emph{complexity increases}).
\item
$\complexity(\phi')=\complexity(\phi)$ and $\f{atomic}(\phi)\gneq\f{atomic}(\phi')$ (in words: \emph{atomic reducts decrease}).
\end{itemize*}
Similarly for $t\to t'$.
\end{corr}
\begin{proof}
Consider a reduct $s\tin\tst{a}{\phi'}$.
\begin{itemize*}
\item
If $s$ is not an atom then we use Lemma~\ref{lemm.complexity.geq}(2).
\item
If $s$ is an atom then using Lemma~\ref{lemm.complexity.geq}(1) complexity is unchanged; however the number of atomic reducts decrements.
\qedhere\end{itemize*}
\end{proof}

\begin{prop}
\label{prop.termination}
The rewrite system from Figure~\ref{fig.rewrite} is terminating (no infinite chain of rewrites).
\end{prop} 
\begin{proof}
We consider just the case of reducing formulae; reducing terms is no harder.

$\f{na}(\phi)$ is just an annotated copy of $\phi$, so if $\phi$ has an infinite chain of rewrites then so must $\f{na}(\phi)$.\footnote{In a machine implementation we would probably want to refine Definition~\ref{defn.pad} to
annotate with $a'\tin c\tand a'\tin c\tand a'\tin c$ for $c$ a fresh constant-symbol or variable symbol (one for each level), instead of $\texi c.(a'\tin c\tand a'\tin c\tand a'\tin c)$.  This would make it easier to automatically track the annotations.} 
We see that it would suffice to prove that reductions from $\f{na}(\phi)$ are terminating.

So suppose $\f{na}(\phi)=\phi$, and $\phi$ is ternary (apply $\f{na}$ if required).

Consider $\phi'$ and suppose $\phi\to^* \phi'$. 
From Theorem~\ref{thrm.confluent} $\phi'\to^*\finte{\phi}$, and by Corollary~\ref{corr.affine.2}(1) 
$$
\complexity(\phi)\leq\complexity(\phi')\leq\complexity(\finte{\phi}).
$$
Thus the set $\{\complexity(\phi')\mid \phi\to^*\phi'\}$ 
is bounded above by $\complexity(\finte{\phi})$. 
It follows using Corollary~\ref{corr.affine.2}(1\&2) and considering the measure 
$$
\bigl(\complexity(\finte{\phi})-\complexity(\phi'),\f{atomic}(\phi')\bigr) ,
$$ 
lexicographically ordered, that any chain of reductions from $\phi$ must terminate. 
\end{proof}

\begin{rmrk}
The proof of Proposition~\ref{prop.termination} is not difficult.\footnote{\dots but not trivial.  Thanks to an anonymous referee for spotting my errors.}
This is in itself interesting:

We noticed in Remark~\ref{rmrk.lambda} how stratified syntax can be viewed as a fragment of the simply-typed $\lambda$-calculus, where $t\in\{a\mid\phi\}$ corresponds to a $\beta$-reduct and extensionality $s=\{b\mid b\in s\}$ corresponds to an $\eta$-expansion.
Yet, the direct proof of strong normalisation for the simply-typed $\lambda$-calculus is quite different and seems harder than the proof of Proposition~\ref{prop.termination} (a concise but clear presentation is in Chapter~6 of \cite{girard:prot}). 
\end{rmrk}

\begin{thrm}
\label{thrm.tst.strong}
Formulae and terms of Stratified Sets, with the rewrites from Figure~\ref{fig.rewrite}, are confluent and strongly normalising. 
\end{thrm}
\begin{proof}
Confluence and \emph{weak normalisation} (every formula/term has some rewrite to a normal form) are from Theorem~\ref{thrm.confluent}. 

\emph{Strong normalisation} follows from weak normalisation and termination (Proposition~\ref{prop.termination}).
\end{proof}

\begin{rmrk}
So from Theorems~\ref{thrm.tst.strong} and~\ref{thrm.capasn}
we see that:
\begin{enumerate*}
\item
the syntax of formulae and terms has normal forms, and furthermore
\item
normal forms with the natural substitution action given by substitute-then-renormalise, corresponds precisely to the theory of internal predicates and terms from Definition~\ref{defn.cred} and~\ref{defn.sigma}, and furthermore
\item
this theory of normal forms is an instance of the notion of nominal algebras for substitution, also called sigma-algebras, as used in the previous literature studying $\lambda$-calculus, first-order logic, and pure substitution \cite{gabbay:repdul,gabbay:semooc,gabbay:capasn-jv}.
\end{enumerate*}
\end{rmrk}

Recall from Remarks~\ref{rmrk.lanss} and~\ref{rmrk.ambiguity} that in stratifiable syntax, as used in Quine's NF, variables do not have predefined levels but we insist on a \emph{stratifiability} condition that $\phi$ and $s$ are only legal if we \emph{could} assign levels to their variables to stratify them.
We obtain as an easy corollary:
\begin{thrm}
\label{thrm.nf}
Formulae and terms of stratifiable syntax, with the rewrites from Figure~\ref{fig.rewrite}, are confluent and strongly normalising. 
\end{thrm}
\begin{proof}
The result follows from Theorem~\ref{thrm.tst.strong} by taking a stratifiable $\phi$, and stratifying it so that we now have $\phi'$ in the language of Typed Sets.
Rewrites on $\phi'$ clearly correspond 1-1 with rewrites on $\phi$, since Figure~\ref{fig.rewrite} makes no reference to the levels of variables.
\end{proof}

\section{Conclusions and future work}
\label{sect.conclusions}

Stratified Sets occupy a nice middle ground between ZF sets and simple types.
They typically appear used as a foundational syntax.
However, we have seen in this paper that Typed-Sets-the-syntax in and of itself forms a well-behaved rewrite system, and a well-behaved nominal algebra. 
This had not previously been noted, and this paper gives a reasonably full and detailed account of how rewriting and nominal algebra apply. 
This account is intended to be suitable for  
\begin{itemize*}
\item
readers familiar with rewriting who are unfamiliar with stratified sets syntax\footnote{Stratified sets syntax is not hard to define --- but it requires experience to learn what kinds of predicates are and are not stratifiable.  In use, stratifiability is a subtle and powerful condition.}
\item
readers familiar with stratified sets syntax but unfamiliar with techniques from rewriting and nominal algebra.
\end{itemize*}

We have also tried to smooth a path to implementing these proofs in a machine, hopefully in a nominal context.
We have designed the proofs to be friendly to such an implementation as future work, yet without compromising readability for humans.
Where we have cut corners (relative to a machine implementation), we tried to signpost this fact (see for instance Remark~\ref{rmrk.slip.in} and Notation~\ref{nttn.inclusive}).

Concerning other applications, it is often possible to use normal forms to build denotations.
In some contexts, the normal form \emph{is} the denotation of the terms that reduce to it.  That will not work for Stratified Sets because we are usually interested in imposing additional axioms.  But there are standard things that can be done about that,
and this has been investigated in a nominal context in papers like \cite{gabbay:semooc,gabbay:repdul}.
These papers build denotations for first-order logic and the $\lambda$-calculus using maximally consistent sets, and using nominal techniques to manage binding in denotations (extending how we used nominal techniques in this paper to manage binding in syntax). 
Having normal forms is useful here and the ideas in this paper can be used to give a denotational analysis of theories in the languages of Stratified and Stratifiable Sets.
This is future work.

We can ask about a converse to Theorems~\ref{thrm.tst.strong} and~\ref{thrm.nf}.
We have shown that a stratifiable formula rewrites to a normal form.
Now if a formula (without levels) rewrites to normal form, is it stratifiable?
We see that we cannot hope for a perfect converse by the following easy example: if we write $\varnothing=\{a\mid \bot\}$ then $\varnothing\in\{a\mid a\not\in a\}$ is not stratifiable but it rewrites to $\varnothing\not\in\varnothing$, which is stratifiable for instance as $\{a^0\mid\bot\}\not\in\{a^1\mid\bot\}$, which we could also write just as $\varnothing^1\not\in\varnothing^2$.
However there may be special cases in which stratification information can be recovered from normalisation, and this is future work. 

\renewcommand\href[2]{#2}

\hyphenation{Mathe-ma-ti-sche}

\end{document}